\documentclass[aps,prb,showpacs,notitlepage,floatfix,twocolumn,superscriptaddress]{revtex4-1}

\usepackage{bbm}
\usepackage{mathrsfs}
\usepackage{subcaption}
\usepackage{epsfig}

\usepackage{soul,xcolor}
\usepackage{graphicx}
\graphicspath{{./figures/}}

\usepackage{amsfonts}
\usepackage{amsthm}
\usepackage[figuresright]{rotating}
\usepackage{amssymb}
\usepackage{amsmath}
\usepackage{dcolumn}
\usepackage{float}
\usepackage{bm}
\usepackage{braket}
\usepackage[normalem]{ulem}
\DeclareMathOperator{\Tr}{Tr}
\usepackage[colorlinks,linkcolor=blue,anchorcolor=blue,citecolor=blue,urlcolor=blue]{hyperref}

\usepackage{physics}

\def\be{\begin{equation}} \def\ee{\end{equation}}
\def\bea{\begin{equation}\begin{aligned}} \def\eea{\end{aligned}\end{equation}}

\newcommand{\symproj}[2][2]{P^{#1,\,#2}_{\text{sym}}}
\newcommand{\symprojT}[2][2]{T^{#1,\,#2}_{\text{sym}}}

\definecolor{KKgreen}{RGB}{0,100,0}

\begin{document}

\title{Deformations of the boundary theory of the square lattice AKLT model}

\author{John Martyn}\affiliation{Department of Physics, University of Maryland, College Park, MD 20742, USA}\affiliation{Institute for Quantum Information and Matter, California Institute of Technology, Pasadena, CA 91125, USA}
\author{Kohtaro Kato}\affiliation{Institute for Quantum Information and Matter, California Institute of Technology, Pasadena, CA 91125, USA}
\author{Angelo Lucia}
\affiliation{Institute for Quantum Information and Matter, California Institute of Technology, Pasadena, CA 91125, USA}
\affiliation{Walter Burke Institute for Theoretical Physics, California Institute of Technology, Pasadena, CA 91125, USA}

\begin{abstract}
The 1D AKLT model is a paradigm of antiferromagnetism, and its ground state exhibits symmetry-protected topological order. On a 2D lattice, the AKLT model has recently gained attention because it too displays symmetry-protected topological order, and its ground state can act as a resource state for measurement-based quantum computation. While the 1D model has been shown to be gapped, it remains an open problem to prove the existence of a spectral gap on the 2D square lattice, which would guarantee the robustness of the resource state. Recently, it has been shown that one can deduce this spectral gap by analyzing the model's boundary theory via a tensor network representation of the ground state. In this work, we express the boundary state of the 2D AKLT model in terms of a classical loop model, where loops, vertices, and crossings are each given a weight. We use numerical techniques to sample configurations of loops and subsequently evaluate the boundary state and boundary Hamiltonian on a square lattice. As a result, we evidence a spectral gap in the square lattice AKLT model. In addition, by varying the weights of the loops, vertices, and crossings, we indicate the presence of three distinct phases exhibited by the classical loop model.
\end{abstract}

\maketitle

\section{Introduction}
The boundary state of a quantum system encodes ample information about the system's bulk. This concept has motivated numerous research developments in different areas, from the theory of topological insulators~\cite{1999IJTP...38.1113M, PhysRevLett.95.226801, PhysRevLett.95.146802} to the AdS/CFT correspondence. 

A correspondence between a system's boundary and its bulk is explicitly manifest in  Projected Entangled Pair States (PEPS)~\cite{2017JPhA...50v3001B, 2014AnPhy.349..117O}. This type of tensor network state has risen to prominence in the past decade due to its versatility in approximating ground states and expressiveness in capturing essential physical properties of exotic phases of quantum matter.

The correspondence we discuss involves the spectral gap of a quantum many-body Hamiltonian, or the difference between its two lowest energy levels. Quantum phase transitions correspond to points in the phase diagram where the spectral gap vanishes. While the problem of determining the presence of a spectral gap or the related phase diagram is in general undecidable~\cite{Cubitt2015,1810.01858,1910.01631}, it is tractable for some specific cases of interest.

In the case of PEPS models, solid numerical evidence~\cite{2011PhRvB..83x5134C} and analytical results~\cite{2017arXiv170907691K, 2004.10516} provide a connection between the spectral gap of the Hamiltonian and locality properties of the boundary state. This bulk-boundary correspondence can be used to evidence a spectral gap in systems that are otherwise difficult to study analytically. 

One noteworthy system is the 2D generalization of the model proposed by Affleck, Kennedy, Lieb and Tasaki (AKLT)~\cite{1987PhRvL..59..799A, 1988CMaPh.115..477A}, whose ground state is known as a valence bond state (VBS), or equivalently, a PEPS. Since the 2D AKLT VBS can be used as a resource for measurement-based quantum computation~\cite{2008PhRvL.101a0502B, 2012PhRvA..86c2328W, 2015arXiv150802595Z}, the spectral gap has implications on the complexity of preparing the states and their stability against noise. However, despite extensive analyses of the 2D AKLT model, its spectral gap remains to be fully understood. Various numerical studies have suggested that the general AKLT model is gapped, some of which have explicitly calculated the spectral gap for small systems and estimated this gap in the thermodynamic limit~\cite{2013PhRvB..88x5118G, 2015PhRvB..92t1111V}. Very recently, the spectral gap for the 2D hexagonal model has been confirmed by numerically verifying a rigorous finite size criteria~\cite{1910.11810}.  

In this work, we numerically study the spectral gap of the 2D AKLT model on a square lattice via the bulk-boundary correspondence in the PEPS framework. For the AKLT model, the boundary states are mapped to a classical loop model, where vertices and loops are each assigned a specific statistical weight. The corresponding classical loop model on the hexagonal lattice has been studied in depth, but not within the context of the AKLT model~\cite{LoopModelNotes}. However, unlike loops on a hexagonal lattice, loops can intersect on a square lattice, and therefore the analysis is more involved in our scenario.

We use random sampling to numerically construct the boundary state and the boundary Hamiltonian on lattices of different sizes. We find evidence of locality of the boundary Hamiltonian and therefore indicate the existence of a spectral gap, in accordance to previous findings in the literature \cite{2011PhRvB..84x5128L, 2010JPhA43y5303K, 2011PhRvB..83x5134C}.
%Because we simulate a classical model which does not suffer from a sign problem, we are able to reach system sizes larger than those studied in the past
We obtain further evidence that the boundary of the AKLT model does indeed show the locality features of a gapped model by varying the parameters of the classical loop model. By doing so, we obtain a phase diagram that exhibits multiple phases. The point corresponding to the AKLT model sits well inside a non-critical phase, thus confirming that the model is gapped. While many deformations of the AKLT model have been proposed in the literature \cite{2011PhRvB..83x5134C, 1711.00036, 1605.08417}, to the best of our knowledge these deformations do not give rise to boundary states related to these loop models, which appear instead for some recently proposed AKLT models on decorated lattices \cite{2019arXiv190109297A, 1905.01275}.

\section{Background}
\subsection{Projected Entangled Pair States (PEPS)} \label{sec:PEPS}

Adopting the conventions of Ref. \cite{2017arXiv170907691K}, PEPS are defined as follows. First begin with a subset $\Lambda$ of an infinite graph on which to model a state. Associate with each vertex $v \in \Lambda$ a $d$-dimensional Hilbert space $\mathcal{H}_d$, known as the physical Hilbert space. The total Hilbert space of the graph is then $\mathcal{H}_\Lambda = \bigotimes_{v \in \Lambda} \mathcal{H}_d$. Next, denote the set of edges of $\Lambda$ by $E_\Lambda$. Associate with each end point of an edge a Hilbert space $H_D$ of dimension $D$, spanned by an orthonormal basis $\{\ket{j}\}_{j=1}^D$. $D$ is known as the bond dimension of the PEPS. On the Hilbert space of edge $e$, denote the maximally entangled state by $\ket{\omega}_e = \frac{1}{\sqrt{D}} \sum_{j=1}^D \ket{j, j}_e$. Focus now on a vertex $v$ of degree $r$, and associate with this vertex a linear map $T_v: \mathcal{H}_D^{\otimes r} \to \mathcal{H}_d$. $T_v$ maps from states of the virtual Hilbert space, $\mathcal{H}_D^{\otimes r}$, to the physical Hilbert space, $\mathcal{H}_d$. 

With these definitions, a PEPS state on $\Lambda$ (with no outgoing edges) is defined, up to normalization, as 
\begin{equation}\label{eq:PEPS}
\ket{\mathrm{PEPS_\Lambda}} = \bigotimes_{v \in \Lambda} T_v \bigotimes_{e\in E_\Lambda} \ket{\omega}_e.
\end{equation}
We interpret this as a state where the transformation $T_v$ is applied to entangled pairs at the edges of each vertex, producing a state of the physical Hilbert space. If the bond dimension $D$ grows as a polynomial in the system size, then this PEPS description is called efficient.
If $\Lambda$ has outgoing edges, we collect the non-contracted virtual degrees of freedom into a boundary Hilbert space $\mathcal{H}_{\partial \Lambda}$. 

\subsection{The Bulk-Boundary Correspondence}
As explained in Refs. \cite{2011PhRvB..83x5134C, BBPEPS, 2011PhRvB..84x5128L}, PEPS exhibit a bulk-boundary correspondence, wherein the virtual degrees of freedom at the boundary correspond to the physical degrees of freedom in the bulk. Explicitly, the boundary theory is governed by the boundary Hamiltonian, defined as follows. First, focus on a region $\Lambda$ with all the outgoing edges. The density matrix on the boundary of $\Lambda$ is then obtained by tracing out the bulk physical degrees of freedom and keeping only the remaining boundary virtual degrees of freedom:
\begin{equation}
\rho_{\partial \Lambda} = \mathrm{tr}_{\Lambda}\big(\dyad{\mathrm{PEPS}_{\Lambda}}\big) \in \mathcal{B}(\mathcal{H}_{\partial \Lambda}), 
\end{equation}
where $\mathcal{B}(\mathcal{H}_{\partial \Lambda})$ is the set of bounded operators on $\mathcal{H}_{\partial \Lambda}$. The boundary Hamiltonian, $H_{\partial \Lambda}$, is defined by equating $\rho_{\partial \Lambda}$ to a thermal state at inverse temperature $\beta = 1$:
\begin{equation} \label{eq:Ham}
\rho_{\partial \Lambda} = e^{-H_{\partial \Lambda}}.
\end{equation}

Via the bulk-boundary correspondence, many properties of the bulk can be learned from $\rho_{\partial \Lambda}$. For instance, the boundary Hamiltonian exhibits the same symmetries as the bulk Hamiltonian. More importantly, the locality of the boundary Hamiltonian indicates a spectral gap in the bulk Hamiltonian \cite{2011PhRvB..83x5134C}. This correspondence is also useful for numerical analyses. Ref. \cite{2011PhRvB..83x5134C} employs the bulk-boundary correspondence to numerically study the entanglement spectrum of the boundary theory of a modified 2D AKLT model.

\subsection{The AKLT Model and the VBS}
%The versatility of PEPS is well-illustrated by the AKLT model \cite{1987PhRvL..59..799A}. 
Following Ref. \cite{2010JPhA43y5303K}, let us define the AKLT model on a graph $G = (V, E)$. With each vertex $k\in V$, we associate a spin operator $\vec{S}_k$ whose spin value is $S_k = |\vec{S}_k| = \frac{1}{2}\mathrm{deg}(k)$, where $\mathrm{deg}(k)$ is the degree of vertex $k$. The AKLT model is then defined by the following Hamiltonian:
\begin{equation}
H_{\mathrm{AKLT}} = \sum_{\langle k, l \rangle \in E} A(k,l) \pi_{S_k+S_l} (k,l),
\end{equation}
where $A(k,l) \in \mathbb{R}^+$ are (arbitrary) coefficients, and $\pi_{S_k+S_l} (k,l)$ is the projector onto the space of maximal total spin of vertices $k$ and $l$: $S_{k,l} := S_k+S_l$. Explicit formulae for $\pi_{S_k+S_l} (k,l)$ are provided in Ref. \cite{2010JPhA43y5303K}. 
 
The ground state of this model is known as the valence-bond-state (VBS) \cite{2011PhRvB..84x5128L, 2010JPhA43y5303K}. It is the unique ground state, up to degeneracies arising from the states of the boundary vertices. As shown in Refs. \cite{2019arXiv190109297A, 2008AdPhy..57..143V, 2017JPhA...50v3001B}, the VBS in 2 dimensions has an exact PEPS representation of bond dimension $D=2$. This representation has been used to study the properties of the AKLT model on different lattices, such as cylinders and hexagonal lattices \cite{2011PhRvB..83x5134C}.

On a graph $G = (V,E)$, the VBS is obtained by placing at each edge $(k,k')$ in $E$ a singlet state shared between $k$ and $k'$ (so that the virtual vector space at a vertex $k$ is $(\mathbb{C}^2)^{\otimes \text{deg}(k)}$ ), and then projecting the virtual space at each vertex onto its appropriate symmetric subspace, which has dimension $\deg(k)$. Formally, the VBS is defined as 
\begin{equation} \label{eq:VBS}
\ket{\text{VBS}} = \bigotimes_{k\in V} \symprojT{\deg(k)}  \bigotimes_{e\in E} \ket{s}_e,
\end{equation}
where $\symprojT{\deg(k)}:(\mathbb{C}^2)^{\otimes \text{deg}(k)} \to \mathbb{C}^{\deg(k)}$ is the projector onto the symmetric subspace, and $\ket{s} = \frac{1}{\sqrt{2}}\qty(\ket{01}-\ket{10})$ is the singlet state. Note that we will ignore the normalization factor in the rest of this paper.

\subsection{The VBS as a PEPS}
Eq.~\eqref{eq:VBS} is almost in the same form as Eq. \eqref{eq:PEPS}, except that we are contracting against antisymmetric singlet states $\ket{s}$ instead of $\ket{\omega}$.
In order to connect with the existing results in the PEPS literature, one would like to express the VBS as in Eq.~\eqref{eq:PEPS}. This is easily remedied by representing each $\ket{s}$ in terms of $\ket{\omega}$, by applying a multiple of a Pauli matrix:  $\ket{s} = (I \otimes i\sigma^y) \ket{\omega}$. This shows that the VBS is indeed a PEPS.

For bipartite lattices, such as the square or hexagonal lattice, one can then choose to distribute the $i\sigma^y$ matrices performing the rotation from $\ket{\omega}$ to $\ket{s}$ in such a way that either all of the virtual indices at a given site are affected, or none are. With this choice, the tensor $T_v$ representing the VBS state as a PEPS depends on which bipartition the vertex $v$ belongs to: in one case it will be exactly $\symprojT{\deg(v)}$, while in the other it will be given by $\symprojT{\deg(v)} (i\sigma^y)^{\otimes \deg(v)}$. The latter is equal to $U_v \symprojT{\deg(v)}$ for some local unitary $U_v$. When computing the boundary state $\rho_{\partial \Lambda}$,  these local unitaries $U_v$ cancel out under the partial trace.

On the contrary, at the virtual boundary sites, on which we are not taking the trace when computing the boundary state, the unitary matrices $i\sigma^y$ will not cancel out, so that the PEPS boundary state will be equal to
\[ 
     U_{\partial \Lambda} 
     \tr_{\text{bulk}} \qty(\bigotimes_{v \in V} \symprojT{\deg(v)} \bigotimes_{e\in E}\dyad{s}_e \bigotimes_{v \in V} \symprojT{\deg(v) \dag}) U_{\partial \Lambda}^\dag,
\]
where $U_{\partial \Lambda}$ is a tensor product of Pauli $\sigma_y$ acting on the virtual boundary Hilbert space. Since this transformation does not affect the locality properties of the boundary state, we will ignore it and compute $\rho_{\partial \Lambda}$ without it, with the understanding that what we obtain is the correct PEPS boundary state for the VBS up to a local change of basis.

Finally, in order to simplify the notation, we will denote by $\symproj{\deg(k)}$ the composition $\symprojT{\deg(k) \dag} \symprojT{\deg(k)}$, thinking of it as an operator on $(\mathbb{C}^2)^{\otimes \deg(k)}$. With this final simplification, the boundary state we are going to compute is 
\begin{equation}
    \rho_{\partial \Lambda} = 
    \tr_{\text{bulk}} \qty(\bigotimes_{v \in V} \symproj{\deg(v)} \bigotimes_{e\in E}\dyad{s}_e).
\end{equation}

\section{Mapping the Boundary VBS to a Classical Loop Model} \label{sec:map}
In this section, we establish the connection between the boundary state of the AKLT model and a classical loop model. In order to make the description more accessible, we apply our method to the VBS on lattices of increasing complexity: a 1D chain, a 2D hexagonal lattice, and finally a 2D square lattice. The calculations regarding the projection operators $\symproj{\deg(k)}$ are presented in Appendix \ref{sec:projops}.

\subsection{Boundary State of VBS on a 1D Chain}
We begin by looking at the boundary state of the VBS on %the simplest possible graph: 
a 1D chain. We can find the exact boundary state in this case, and it will serve as a precursor to studying the VBS on 2D graphs.

On the 1D chain, $\text{deg}(k) = 2 \ \ \forall k \in V$. Consider an interval $[1,n]$, whose bulk is composed by $n$ spin-1 degrees of freedom, while the boundary consists of the two spin 1/2 degrees of freedom at the outgoing edges.
We trace out the bulk and we denote by $\rho_n$ the resulting boundary state, that is,  
\begin{equation}\label{eq:rhon1d}
\rho_n=\mathrm{tr}_{\rm bulk}\left( \bigotimes_{i=1}^n \symproj{2}\dyad{s}^{\otimes n+1}\right).
\end{equation}

We can write $\symproj{2}$ as (see Appendix \ref{d2n2} for details): 
\begin{equation}
\symproj{2} = \frac{3}{4} \Big(I + \frac{1}{3} \sigma_1 \cdot \sigma_2 \Big),
\end{equation}
where $\sigma_1 \cdot \sigma_2 := \sigma^x \otimes \sigma^x + \sigma^y \otimes \sigma^y + \sigma^z \otimes \sigma^z$, and the sub-indices on $\sigma$ denote the virtual sites on which the operator acts. 
%\KK{Comment:The following paragraphs are a bit hard to understand without equation. (AL): I tried to write down one.}
One can then expand the tensor product $\bigotimes_{i=1}^n \symproj{2}$ as a polynomial in the $\sigma \cdot \sigma$ terms:
\begin{equation}\label{eq:1d-expansion}
    \bigotimes_{i=1}^n \symproj{2} = \qty(\frac{3}{4})^n \sum_{J\subset[1,n]} \frac{1}{3^{\abs{J}}} \bigotimes_{i\in J} \sigma_{i,1}\cdot \sigma_{i,2},
\end{equation}
where the sub-index denotes which of the two virtual sites of each site the operator acts on.
However, it turns out that when contracted against the $\ket{s}$ states in Eq.~\eqref{eq:rhon1d}, all except two of these terms vanish. 

To see this relation and straightforwardly calculate the boundary state, it is convenient to introduce a diagrammatic representation of the expansion in Eq.~\eqref{eq:1d-expansion}. For each term in the expansion in Eq.~\eqref{eq:1d-expansion}, we say a site is ``turned on" if $\sigma\cdot\sigma$ acts on it, and ``turned off" if the identity $I$ acts on it. We then associate a configuration of \emph{strings} to each term by drawing lines on the edges adjacent to every turned on site. A string ends when it reaches a turned off site or one of the boundary sites.

The expansion in Eq.~\eqref{eq:1d-expansion} is then an expansion over all the possible configuration of such strings (where each string corresponds to a connected component of the set $J$). It is easy to show that if a string begins or terminates in the bulk then it necessarily vanishes, since $\expval{I \otimes \sigma^i}{s} = 0$ for $i=x,y,z$.

Therefore, there are only two configurations that contribute to $\rho_n$: the configuration with all sites turned off, which has no strings, and the configuration with all sites turned on, which has a string connecting the two boundary qubits. Neither of these configurations contain strings that begin or terminate in the bulk. 

The configuration with no strings yields a contribution equal to the identity. By using the fact that $\expval{\sigma^i \otimes \sigma^j}{s} = (-1)\delta_{ij}$, one can see that the configuration with all sites turned on yields a contribution of $(-1)^{n+1} 3^{-n} \sigma_0 \cdot \sigma_{n+1}$ (where $0$ and $n+1$ denote the two boundary spins). Therefore, after normalizing the state, the boundary state of the 1D VBS in 1D is
\begin{equation}
\rho_n = \frac{I}{4} + \frac{(-1)^{n+1}}{3^n} \frac{\sigma_0 \cdot \sigma_{n+1}}{4}. 
\end{equation}

\subsection{Boundary State of VBS on a Hexagonal Lattice}\label{VBS_hex}
We now turn to the VBS on a 2D hexagonal lattice, which we can study analogously to the 1D case. On this lattice, $\text{deg}(k) = 3$, so we will use (see Appendix \ref{d2n3} for details)
\begin{equation}
\symproj{3} = \frac{1}{2} \Big( I + \frac{1}{3}(\sigma_1 \cdot \sigma_2 + \sigma_2 \cdot \sigma_3 + \sigma_1 \cdot \sigma_3)\Big),
\end{equation}
where again the subindex denotes the virtual leg that the operator acts on. In this case, there are three types of turned on sites corresponding to different combinations of outgoing lines determined by the subindices of $\sigma$'s. We construct a diagrammatic representation by drawing lines only on the outgoing edges of each site where the $\sigma$'s are acting.
By using arguments similar to those in the 1D case, we find that the only contributions to $\rho_\partial$ come from configurations where no strings begin or terminate in the bulk. Thus, $\rho_\partial$ is a sum over configurations where strings either form closed loops in the bulk, or terminate at the boundary of the hexagonal lattice. Since each vertex has degree 3 and the number of $\sigma$'s in the expansion at each site has to be even, no crossings of the loops or strings are allowed. 

As the only terms contributing to the boundary state are configurations of loop and strings with endpoints on the boundary, we can collect all the terms with the same endpoints, so that $\rho_\partial$ will be a weighted sum of products of terms $\sigma_i \cdot \sigma_j$, where $i$ and $j$ are virtual sites at the boundary. Let $\mathbf{i} = \{ (i_1, j_1), (i_2, j_2), ..., (i_k, j_k) \}$ be a set of disjoint pairs of boundary qubits and ${\mathcal I}=\{\mathbf{i}\}$ be the set of all possible such disjoint pairs.  The boundary state is then given by
\begin{equation}\label{eq:rho_partial}
\begin{gathered}
\rho_\partial = \frac{1}{\mathcal{N}}\sum_{\textbf{i}\in\mathcal{I}} Z^{\textbf{i}} \prod_{k=1}^n \sigma_{i_k} \cdot \sigma_{j_k},
\end{gathered}
\end{equation}
where $\{Z^{\mathbf{i}}\}_{\mathbf{i}}$ are coefficients and $\mathcal{N}$ is a normalization factor. Each coefficient is obtained by summing the contributions from all possible loops and string configurations with endpoints $\mathbf{i}$, which we denote by $\mathcal{C}_{\mathbf{i}}$:
\begin{equation}
Z^{\textbf{i}} = \sum_{\Delta \in \mathcal{C}_{\mathbf{i}}} w(\Delta).
\end{equation}
Here, $w(\cdot)$ is the weight function determined by the contraction of the tensors, which can be computed with the following rules, starting from the empty diagram with no strings which is given an arbitrary weight (the choice of which will only affect the normalization of the boundary state, so we set it equal to 1 without loss of generality):
\begin{enumerate}
%\item The empty configuration with no loops and no strings, corresponding to the choice of $I$ at each site, has weight $1$;\KK{Do we need this trivial step? }
\item Each edge contributes a multiplicative factor of $(-1)$;
\item Each loop contributes a multiplicative factor of $3$;
\item Each vertex with a string passing through it contributes a multiplicative factor of $1/3$.
\end{enumerate}
For a given configuration $\Delta$, if we denote by $m$ the total number of edges covered by strings, $\ell$ the number of loops, and $v_2$ the number of bulk vertices with a string passing through them (we choose the index 2 because this is the degree of the vertex), then we obtain the following weight, which we can interpret as a Boltzmann factor up to a sign:
\begin{equation}
w(\Delta) =(-1)^{m} 3^{-v_2 + \ell} = (-1)^{m} \exp[-\ln(3) (v_2  - \ell)].
\end{equation}
On a bipartite graph like the hexagonal lattice, the weights of configurations with the same endpoints all have the same sign, so the sign of these weights only affects the overall sign of $Z^{\textbf{i}}$ without causing any cancellation. Therefore, up to the sign, each coefficient $Z^{\textbf{i}}$ can be interpreted as a partition function of a classical loop model at inverse temperature  $\beta =1 $ with boundary conditions fixed by $\textbf{i}$, and the energy of a configuration given by $\ln(3) (v_2 - \ell)$. This has been previously analyzed as a classical loop model on a hexagonal lattice, albeit not in connection to the AKLT model \cite{LoopModelNotes}.
%\KK{I want to mention about the previous work on the loop model here.}

It is also now straightforward to calculate the normalization factor $\mathcal{N}$ in Eq. (\ref{eq:rho_partial}). Because $\mathrm{tr}(\sigma) = 0$,  $\mathcal{N}$ is determined entirely by the coefficient of the identity in $\rho_\partial$, which corresponds to configurations with no strings connecting boundary qubits. Denoting this coefficient by $Z^{\emptyset}$ (we use the empty set to signify the identity), and letting $N_\partial$ be the number of boundary qubits, we see that $\mathcal{N} = 2^{N_\partial}Z^{\emptyset}$.

\subsection{Boundary State of VBS on a 2D Square Lattice}
Our goal is to study the VBS on a 2D square graph. We have now built up enough formalism to approach this problem.

On a 2D square graph, $\text{deg}(k) = 4$, so the relevant projector is now (see Appendix \ref{d2n4} for details)
\begin{equation}
\begin{split}
\symproj{4} = \frac{15}{48} \Bigg(I \ + \  \frac{1}{3}\sum_{(i,j)\in \mathcal{S}_4}(\sigma_i \cdot \sigma_j) \ \\ +\ \frac{1}{15} \sum_{(i,j)(k,l)\in \mathcal{S}_4} (\sigma_i \cdot \sigma_j) (\sigma_k \cdot \sigma_l) \Bigg).
\end{split}
\label{eq:projector-square}
\end{equation}
Here, $\mathcal{S}_4$ is the set of permutations of 4 elements, $\sum_{(ij) \in \mathcal{S}_4}$ denotes the sum over all $2$-cycles in $\mathcal{S}_4$ (there are 6 of these), and $\sum_{(ij)(kl) \in \mathcal{S}_4}$ denotes the sum over all disjoint $(2,2)$-cycles in $\mathcal{S}_4$ (there are 3 of these).

Evidently, this projector contains both 2-body and 4-body interactions. Similar to the cases on the 1D chain and the hexagonal lattice, we can again find $\rho_\partial$ and expand it in terms of $\sigma_{i_k}\cdot \sigma_{j_k}$ terms as in Eq. (\ref{eq:rho_partial}). The coefficients of these terms are again sums over allowed configurations of loops and strings. However, due to the presence of the 4-body term in $\symproj{4}$, the loops and strings are now allowed to cross. In addition, there is now an interaction term between loops and strings passing through the same vertex, irrespectively of whether they cross or not: they are weighted by $1/15$ instead of the weight of $(1/3)^2 = 1/9$ that they would receive as the product of independent strings.

Analogous to the rules on the hexagonal lattice, the rules for calculating the weight $w(\Delta)$ of a configuration on the square lattice are given by:
\begin{enumerate}
%\item The empty configuration with no loops and no strings, corresponding to the choice of $I$ at each site, has weight $1$;
\item Each edge contributes a multiplicative factor of $(-1)$;
\item Each loop contributes a multiplicative factor of $3$;
\item Each vertex with a single string passing through it contributes a multiplicative factor of $1/3$;
\item Each vertex with two single string passing through it contributes a multiplicative factor of $1/15$.
\end{enumerate}
If we denote by $m$ the total number of edges covered by strings, $\ell$ the number of closed loops, $v_2$ the number of bulk vertices with only one string passing through, and $v_4$ the number of vertices with two strings passing through (again we choose indices 2 and 4 because these are the degrees of the respective vertices), then the weight of an arbitrary configuration is given by the Boltzmann weight:
\begin{equation}
\begin{gathered}
w(\Delta) =  (-1)^m 3^{-v_2 + \ell} (15)^{-v_4}  =\\ (-1)^m \exp[-\ln(3)( v_2 - \ell) -\ln(15) v_4].
\end{gathered}
\end{equation}
As before, because of the bipartite nature of the lattice, the sign $(-1)^m$ only depend on the boundary sites $\mathbf{i}$, so that we can interpret the coefficients in $\rho_\partial$ as partition functions of a classical loop model. As an example, in Fig. \ref{fig:square_config} we display a sample configuration and calculate its weight.

\begin{figure}[ht]
\center{\includegraphics[width=.49\textwidth]{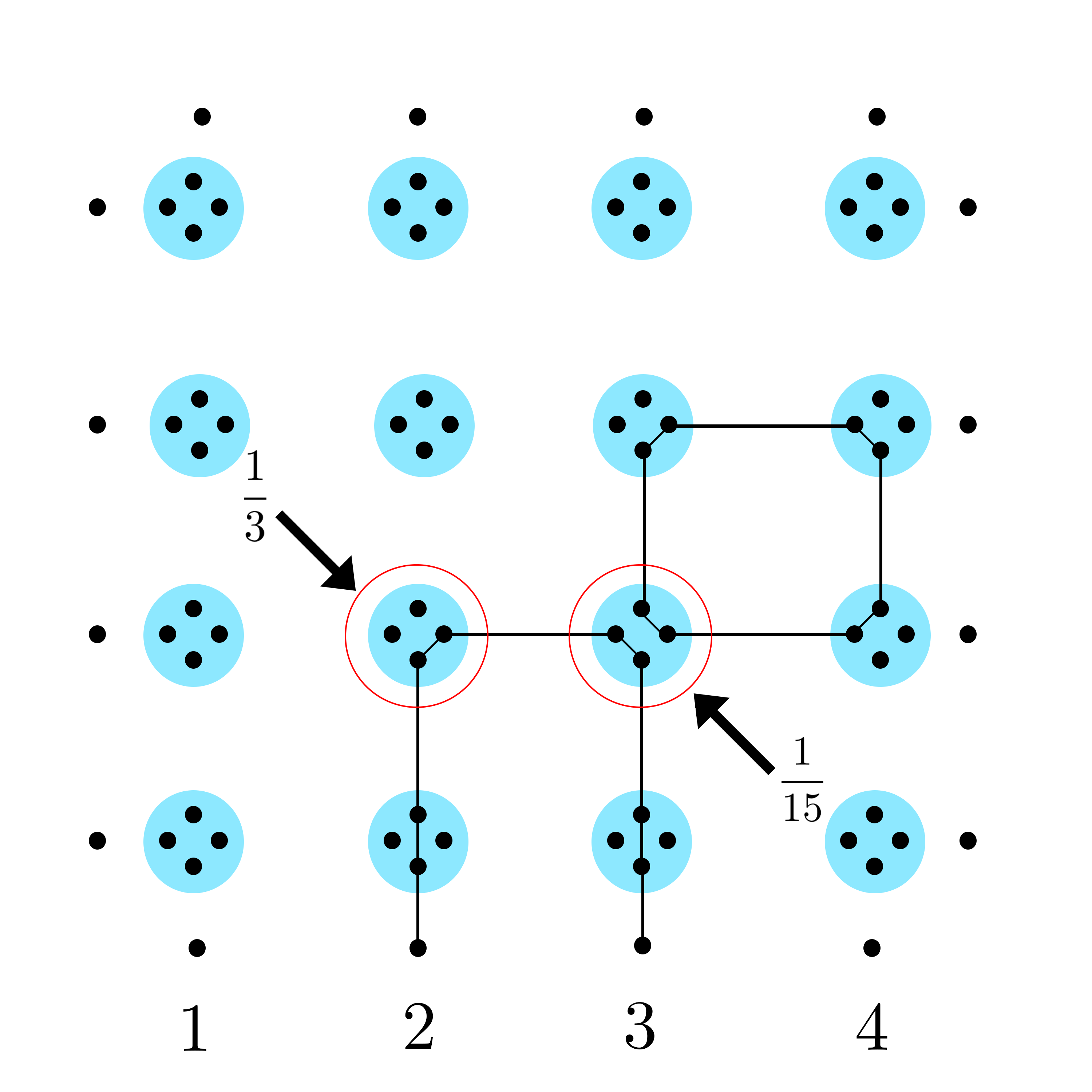}}
\caption{An example configuration on a square lattice. The circled vertex with degree 2 (one string passing through it) contributes a multiplicative factor of $\frac{1}{3}$, whereas the circled vertex with degree 4 (two strings passing through it) contributes a factor of $\frac{1}{15}$. We can compute its weight straightforwardly: we have $m = 9$, $\ell = 1$, $v_2 = 6$, and $v_4 = 1$. Hence the weight of this configuration is $- (3^5 \cdot 15)^{-1}$. Because the strings of this configuration connect boundary spins $2$ and $3$, this weight contributes to $Z^{\{(2,3)\}}$.}
\label{fig:square_config}
\end{figure}

\subsection{Locality at the Boundary} \label{sec:local}
Our goal is to determine the locality of the boundary Hamiltonian, $H_\partial$, on the 2D cylindrical square lattice, which will evidence the presence of a spectral gap. As we describe in Appendix \ref{sec:num}, we study cylindrical lattices with $N_x$ plaquettes in the periodic direction and $N_y$ plaquettes in the longitudinal direction. We use the numerical procedure outline in Appendix \ref{sec:num} to estimate $\rho_\partial$ and subsequently compute the boundary Hamiltonian as $H_\partial = - \log (\rho_\partial)$. 

\subsubsection{Probing Locality}

To probe the locality of the boundary Hamiltonian, we look at the amplitude of the Heisenberg interactions:
\begin{equation}
A_r = \frac{1}{3 \cdot 2^{N_x}} \Tr \big(H_\partial \  \sigma_i \cdot \sigma_{i+r}\big)
\end{equation}
(this is the same for all $i$ due to the rotational symmetry of the cylinder). This amplitude measures the strength of interactions between two qubits separated by $r$ lattice spacings. It is the coefficient of the Heisenberg term in the boundary Hamiltonian:
\begin{equation}
H_\partial = \sum_r A_r \sum_i \sigma_i \cdot \sigma_{i+r} \ + \  ...
\end{equation} 
In addition, we will study the strength of interactions on subsets of n qubits:
\begin{equation}
d_n = \Tr(h_n^2),
\end{equation}
where $h_n$ contains the terms in $H_\partial$ with interaction range $n$. As defined in Ref. \cite{2011PhRvB..83x5134C}, these are the terms in which the largest contiguous block of identity operators acts on $N_x - n$ qubits. For instance, $h_0$ consists of the term in $H_\partial$ that is proportional to the identity, $h_2$ consists of all Heisenberg interactions on neighboring qubits, and so on. Because $H_\partial$ is comprised of products of Heisenberg interactions, which can be deduced from Eq. (\ref{eq:rho_partial}), $d_1$ = 0. All other $d_n$ are nonzero.

Similar quantities are analyzed in Ref. \cite{2011PhRvB..83x5134C} in order to study the locality of a boundary Hamiltonian. For a system of local interactions, we expect both $A_r$ and $d_n$ to decay exponentially in their respective argument or vanish past a certain argument value. We will use these features to diagnose locality.

\subsubsection{Results}

In Fig. \ref{fig:amp_dist_10by10}, we plot $A_r$ and $d_n$ on a lattice of size $N_x=N_y=10$. We see that $A_r$ decays exponentially in $r$, and $d_n$ decays exponentially in $n$. Similar results were seen for lattices of other sizes. Clearly, the interactions in the Hamiltonian decay exponentially in their strength and size, and thus the boundary Hamiltonian is quasi-local. 
%\KK{Would be nice to emphasize again that the system size is bigger than the previous results.}

\begin{figure}[]
\center{\includegraphics[width=.49\textwidth]{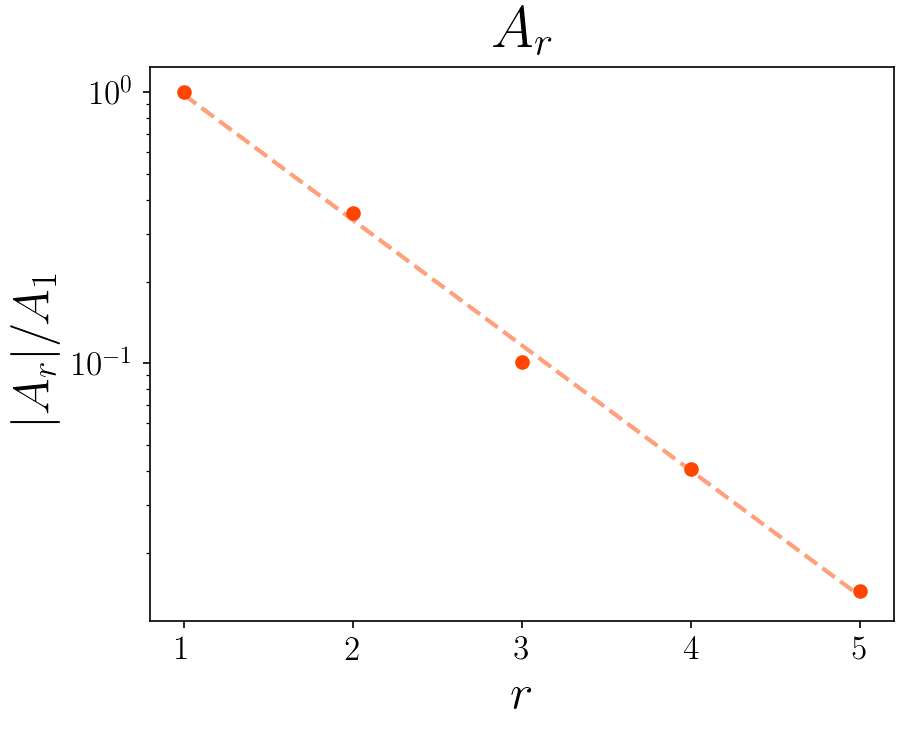}
\includegraphics[width=.49\textwidth]{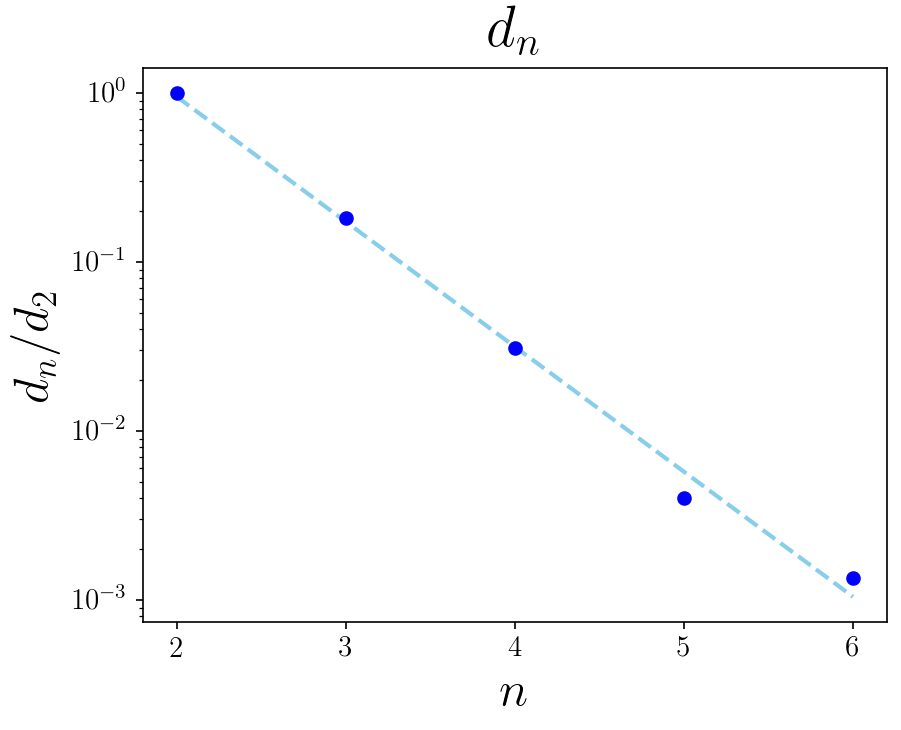}}
\caption{Plots of the Heisenberg interaction amplitudes ($A_r$), the n-qubit interaction strengths ($d_n$), and the corresponding lines of best fit. This data was collected on a lattice of size $N_x=N_y=10$. Also note that these plots only extend to $r=5$ and $n=6$ because the rotational symmetry of the cylinder inhibits exponential decay at larger argument values. For instance, $A_1 = A_9$ because of this symmetry.}
\label{fig:amp_dist_10by10}
\end{figure}

To justify this quasi-locality in the thermodynamic limit, we performed a finite size scaling analysis to extrapolate $A_r$ in the limit of an infinitely long cylinder ($N_y \rightarrow \infty$). The results are shown below in Fig. \ref{fig:FSS}. Note that we collect data on systems up to size $(N_x,\ N_y) = (8,\ 20)$.
% which is larger than those accessible in previous studies~\cite{2011PhRvB..84x5128L, 2010JPhA43y5303K}. 
%\KK{(Comment) Since this  paper argue the merits to study the boundary theory instead of the bulk, I think it would be nice if we can put some reference directly calculating the gap of 2D AKLT models. Unfortunately, the recent gap proof paper does not contain any of such references tho. }

\begin{figure}[]
\center{\includegraphics[width=.49\textwidth]{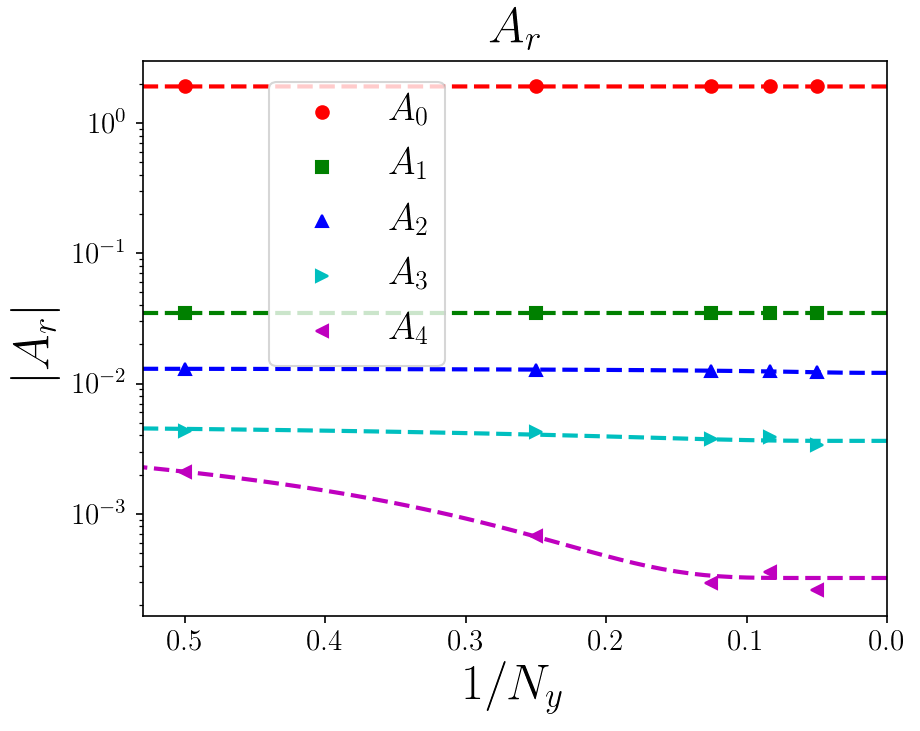}}
\caption{Finite size scaling analysis of Heisenberg interaction amplitudes ($A_r$) as a function of $1/N_y$ on a cylindrical lattice with $N_x=8$. Data is displayed for $N_y = 2,\ 4,\ 8,\ 12,\ 20$. We note that, although $A_0$, $A_1$, $A_2$, and $A_3$ appear constant, they actually grow as a function of $1/N_y$, but less noticeably than $A_4$.}
\label{fig:FSS}
\end{figure}

The finite size scaling analysis indicates that $A_r$ (for $r\geq 1$) continues to decay exponentially in the limit that $N_y \rightarrow \infty$. Hence, the boundary Hamiltonian is quasi-local in the thermodynamic limit. In Ref. \cite{2017arXiv170907691K} it was shown that a strictly local boundary Hamiltonian implies a spectral gap in the bulk, and it was furthermore conjecture that the same results holds for quasi-local interactions, which has been partially confirmed by the recent results of Ref. \cite{2004.10516}. The locality estimate we have obtained thus provides a strong evidence that the 2D AKLT model on a square lattice is gapped.

\section{Phase Structure of the Loop Model}
\subsection{Procedure}
%\KK{(I think it would be better to put phase diagrams first and then discuss about it's relation to the AKLT model. However, it is up to you guys.) Here we investigate the phase structure of the loop model by varying the factors associated with closed loops ($3$), vertices with a single string passing through ($\frac{1}{3}$), and vertices with two strings ($\frac{1}{15}$). We will then discuss the relation of these changes of parameters to perturbations of the AKLT model. (Then just discuss how to study phases, etc..)}
Here we investigate the phase structure of the loop model by varying the factors associated with closed loops, vertices with a single string passing through, and vertices with two strings. These factors were 3, $\frac{1}{3}$, and $\frac{1}{15}$, respectively, in our previous analysis.
%$c_{\ell}=3$, vertices with a single string passing through $c_{2}=\frac{1}{3}$, and vertices with two strings $c_{4}=\frac{1}{15}$.
By varying these parameters we can induce qualitative changes in the boundary state and reveal a phase diagram for the classical loop model.

%Note that if the weight of a closed loop is set to an integer, it is possible to modify Eq. \eqref{eq:projector-square} to obtain a PEPS whose boundary state is given exactly by the modified loop model. Indeed, we can consider the operator
To study the phase structure, we introduce a parameterized operator
\begin{align*}
P(c_2, c_4) &= I  +\, c_2 \sum_{(i,j)\in \mathcal{S}_4}(\sigma_i \cdot \sigma_j)  \\ 
&+\, c_4 \sum_{(i,j)(k,l)\in \mathcal{S}_4} (\sigma_i \cdot \sigma_j) (\sigma_k \cdot \sigma_l),
\end{align*}
which when contracted against the network of singlet states $\ket{s}$ produces a boundary state described by a loop model with closed loops weighted by $3$, vertices with a single string passing through by $c_2$, and vertices with two strings by $c_4$. By replacing the operator $\sigma \cdot \sigma$ by a sum over only one or two of the Pauli matrices, we can change the loop weight to $1$ or $2$, respectively. Higher integer loops weight can be obtained similarly by allowing a larger bond dimension and choosing a larger set of orthonormal traceless Hermitian unitary matrices. We denote this loop weight by $c_\ell$. 
%Note that if the weight of a closed loop is set to an integer (and $P(c_2, c_4)$ is positive semi-definite, as we discuss in the next paragraph), it is possible to modify Eq. \eqref{eq:projector-square} to obtain a PEPS whose boundary state is given exactly by the modified loop model.

In general, $P(c_2,c_4)$ is not a projector. Except for a few special points of the parameters (such as the one corresponding to the VBS), the rank of $P(c_2,c_4)$ will be maximal, which implies a growth in the local physical dimension. 

Furthermore, $P(c_2,c_4)$ is not necessarily positive semi-definite either. However, the boundary state computed with $P(c_2, c_4)$ can be obtained from a PEPS contraction only at points where $P(c_2, c_4)$ is positive semi-definite. In this regime, one can interpret a change in $c_2$ and/or $c_4$ as a perturbation of the PEPS tensor of the VBS after renormalizing at the injectivity length. Note that this perturbation does not break the $SU(2)$ symmetry of the tensor. Outside of the range where $P(c_2, c_4)$ is positive semi-definite, the resulting boundary state is not guaranteed to be a valid density matrix. For instance, it may have negative eigenvalues and therefore not be a physical quantum state. However, it is still meaningful to study the classical loop model in this region, as it exhibits an interesting phase structure. 

To probe the behavior of the boundary state when the weights are changed, we study the spin-spin correlation functions $T_k = \langle {\sigma}_{i} \cdot {\sigma}_{i+k} \rangle - \langle {\sigma}_{i} \rangle \cdot \langle {\sigma}_{i+k} \rangle$, where the average is computed from the boundary state. (This definition is independent of $i$ by the rotational symmetry of the cylinder.) At first glance, this quantity may seem ill-defined when the boundary state is not physical, but $T_k$ really only depends on the classical loop model. To see this, note that Eq. (\ref{eq:rho_partial}) implies $\langle \sigma_i \rangle = 0$ and that $\langle {\sigma}_{i} \cdot {\sigma}_{i+k} \rangle = Z^{\{(i, i+k) \}}/Z^{\emptyset}$, where we have used the value of $\mathcal{N}$ mentioned in Sec. \ref{VBS_hex}. Therefore, $T_k = Z^{\{(i, i+k) \}}/Z^{\emptyset}$ is simply the ratio of two classical partition functions, which is well-defined even when the boundary state is not physical.

We also record the average number of loops and the average perimeter of loops, where these averages are weighted means computed with the Boltzmann factors of loop configurations. These loop properties, as well as the spin-spin correlation functions, are calculated using the random sampling procedure detailed in Appendix \ref{sec:num} to sample configurations on the lattice. 

\subsection{Phases}
\subsubsection{Integer $c_\ell$}
We first study $T_k$ at integer $c_\ell$ and various choices of $c_2$ and $c_4$. We find two characteristic phases exhibited by our model: one in which $|T_k|$ decays exponentially, and one in which it decays sub-exponentially, roughly as a power law. In addition, as indicated by our numerics, the phases can be indexed by a simple, but non-local, order parameter given by the ratio 
\begin{equation}
\eta := \frac{\text{Average Number of Loops}}{\text{Average Total Loop Perimeters}}.
\end{equation}
The average number of loops is computed as the weighted average over all configurations of the number of loops in a configuration. The average total loop perimeter is an analogous average of the sum of perimeters of the loops in a configuration. Qualitatively, this ratio measures the size and prevalence of loops. Because both the numerator and denominator are extensive quantities, this ratio can be finite and nonzero in the thermodynamic limit. 

%\KK{as follows: first take an average over loops in a given configuration, and then average over all configurations. The average number of loops is simply an average over all configurations.} Qualitatively, this ratio measures the size and prevalence of loops. As such, only configurations with a nonzero number of loops contribute to $\eta$. 

We find that when the ground state is in the phase of exponentially decaying correlations, $\eta \approx 1/4$. On the other hand, when the ground state is in the phase of power law decaying correlations, $0 < \eta < 1/4$; a typical such value would be $\eta \sim 0.1$. We also discover that $\eta \approx 0$ corresponds to a nonphysical regime, as we discuss in Sec.~\ref{para:growing}.

In Fig. \ref{fig:NewPhasePlot}, we display plots of $\log(\frac{\eta}{1/4})$ for various $c_2$ and $c_4$ at $c_\ell = 3,\ 2,\ 1$. We also enclose in dashed red lines the regime in which the boundary state is physical. We compute this regime by calculating the eigenvalues of $P(c_2, c_4)$ to determine where this operator is positive semi-definite.

\begin{figure}[H]
\center{\includegraphics[width=.4\textwidth]{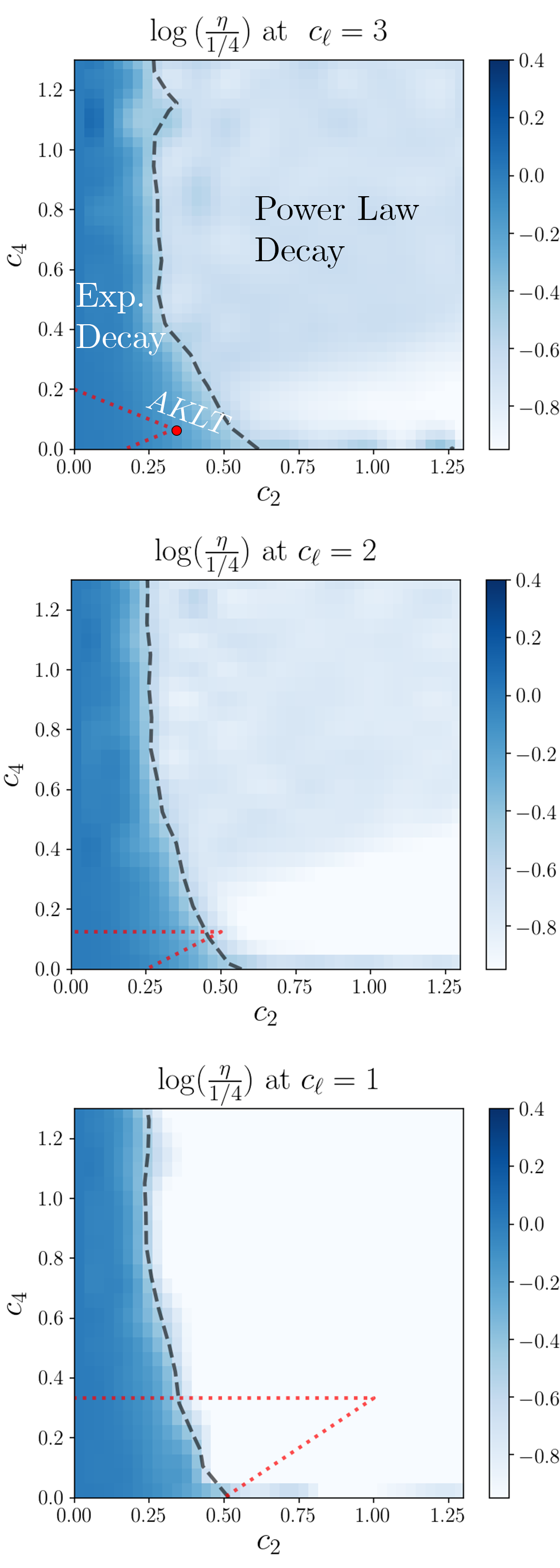}}
\caption{$\log(\frac{\eta}{1/4})$ for various $c_2$ and $c_4$, with integer $c_\ell$. We also plot the contour $\eta = 0.15$ as a dashed gray line to help visualize the phases (other values of $\eta \sim 0.15$ produce similar lines). Likewise, the dashed red line encloses the region where the boundary state is physical, in which $P(c_2, c_4) \geq 0$. \textbf{(Top)} $c_\ell = 3$. This plot contains the AKLT model, which we indicate with a red circle. \textbf{(Middle)} $c_\ell = 2$. \textbf{(Bottom)} $c_\ell = 1$.}
\label{fig:NewPhasePlot}
\end{figure}

The region in Fig. \ref{fig:NewPhasePlot} where $\log(\frac{\eta}{1/4}) \approx 0$ (the blue region) encompasses the phase of exponentially decaying correlations, and the region where $\log(\frac{\eta}{1/4}) < 0$ (the white region) encompasses the phase of power law decaying correlations. From these plots, we see that the AKLT model is well within the phase of exponentially decaying correlations. Increasing $c_2$ and/or $c_4$ will cause the model to enter the phase of power law decaying correlations.

Fig. \ref{fig:NewPhasePlot} also indicates that increasing $c_4$ sharpens the transition between the two phases and shifts the transition to lower values of $c_2$. To study this transition more closely, we calculated $\eta$ as a function of $c_2$ near the transition. We display our results for different lattice sizes in Fig. \ref{fig:eta_transitions}. In the top plot, we fixed $c_\ell = 3$ and $c_4 = 0.2$; in the bottom plot, we fixed $c_\ell = 3$ and $c_4 = 1.2$. It is clear that $\eta$ transitions from $\approx 1/4$ to a value around $\sim 0.1$ more sharply when $c_4 = 1.2$ than when $c_4 = 0.2$. In addition, these plots indicate that, in the limit of an infinite length cylinder, $\eta$ converges to a value near $0.1$ in the power law phase.

\begin{figure}[H]
\center{\includegraphics[width=.45\textwidth]{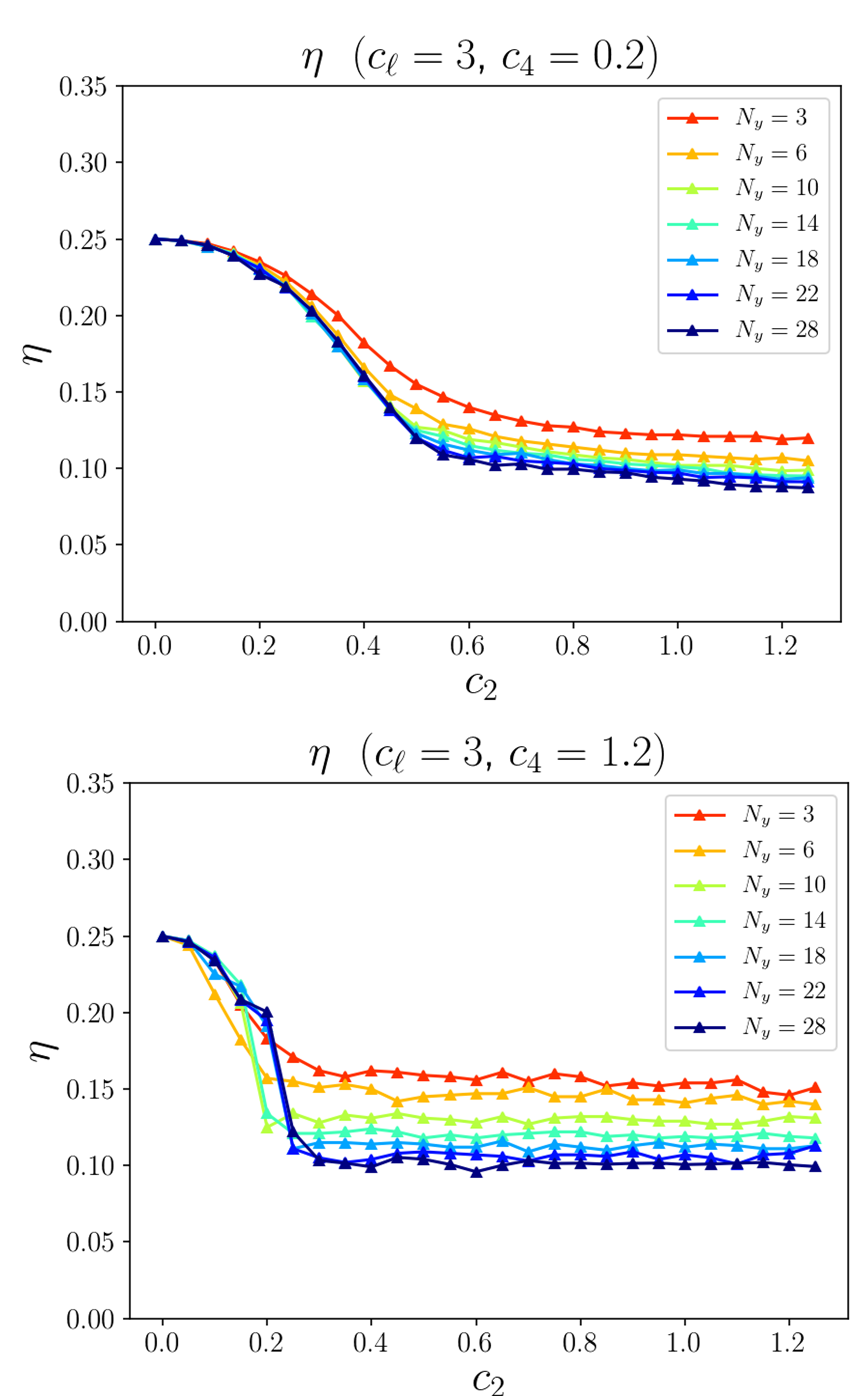}}
\caption{$\eta$ as a function of $c_2$ near the transition. We have fixed $N_x = 10$ and displayed data for various $N_y$. \textbf{(Top)} $c_4 = 0.2$ and $c_\ell = 3$. \textbf{(Bottom)} $c_4 = 1.2$ and $c_\ell = 3$.}
\label{fig:eta_transitions}
\end{figure}

\paragraph{Phase of Exponentially Decaying Correlations} 
The phase of exponentially decaying correlations encompasses regions of parameter space in which $c_2$ and $c_4$ are sufficiently small.
%\KK{MEMO: I'm not sure if the border in FIG.5 reaches $c_2=0$ if $c_4$ is large enough. Intuitively it seems true tho.} 
As its name suggests, this phase is characterized by exponentially decaying correlation functions, as shown in Fig. \ref{fig:T_k_Exp}. In addition, $\eta \approx 1/4$ in this phase. This means that the dominant configurations are those that contain separated loops that each cover one plaquette and have a perimeter of 4. 
%(For clarity, we note that the most dominant configuration in this phase is the configuration with no loops, i.e. $w(0 \ \text{loops}) > w(1 \ \text{loop})$. However, this configuration does not contribute to $\eta$ because it does not contain any loops. Instead, only configurations with a nonzero number of loops, of which the configurations with 1 loop are dominant, contribute to $\eta$. This explains why $\eta$ remains finite, even in the thermodynamic limit.)

The phase of exponentially decaying correlations also encompasses physical states. $P(c_2, c_4)$ is positive semi-definite in some, but not all, of this phase. For instance, the 2D AKLT model falls under this phase, which we depict in Fig. \ref{fig:eta_transitions}. As we demonstrated that this model has a quasi-local boundary Hamiltonian and thus a gapped bulk Hamiltonian, we expect that other physical states in this phase also have boundary and bulk Hamiltonians with the same properties.

\begin{figure}[H]
\center{\includegraphics[width=.49\textwidth]{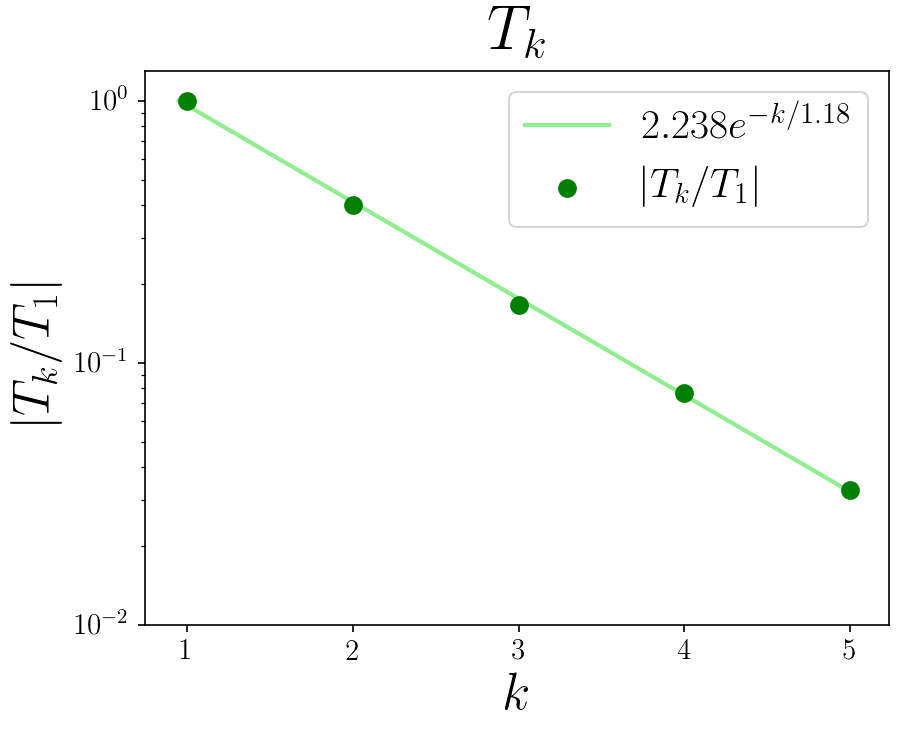}}
\caption{Spin-spin correlation function, $T_k$, in the exponential decaying phase on a lattice of size $N_x=10=N_y$. This data was collected at the AKLT point: $(c_\ell,\ c_2,\ c_4) = (3,\ 1/3,\ 1/15)$.}
\label{fig:T_k_Exp}
\end{figure}

$ $ \\
\paragraph{Phase of Power Law Decaying Correlations}
The loop model enters its phase of power law decaying correlations when $c_2$ and $c_4$ are sufficiently large. This phase is characterized by a spin-spin correlation function that decays sub-exponentially and roughly as a power law, as we display in Fig. \ref{fig:T_k_Pow}. Also, $0 < \eta < 1/4$ in this phase, which indicates that the dominant configurations contain multiple loops of modest size (say, loops that each cover two or three plaquettes).

The phase of power law decaying correlations encompasses physical states, but (roughly) only when $c_\ell = 1$, as we see from Fig. \ref{fig:NewPhasePlot}. Because the correlations in this phase decay slower than exponentially, we expect these states to have boundary Hamiltonians with power law decaying amplitudes and possibly non-local properties. As the proof in Ref. \cite{2017arXiv170907691K} does not apply to such Hamiltonians, we cannot definitively say whether or not the corresponding bulk Hamiltonian is gapped.

\begin{figure}[ht]
\center{\includegraphics[width=.49\textwidth]{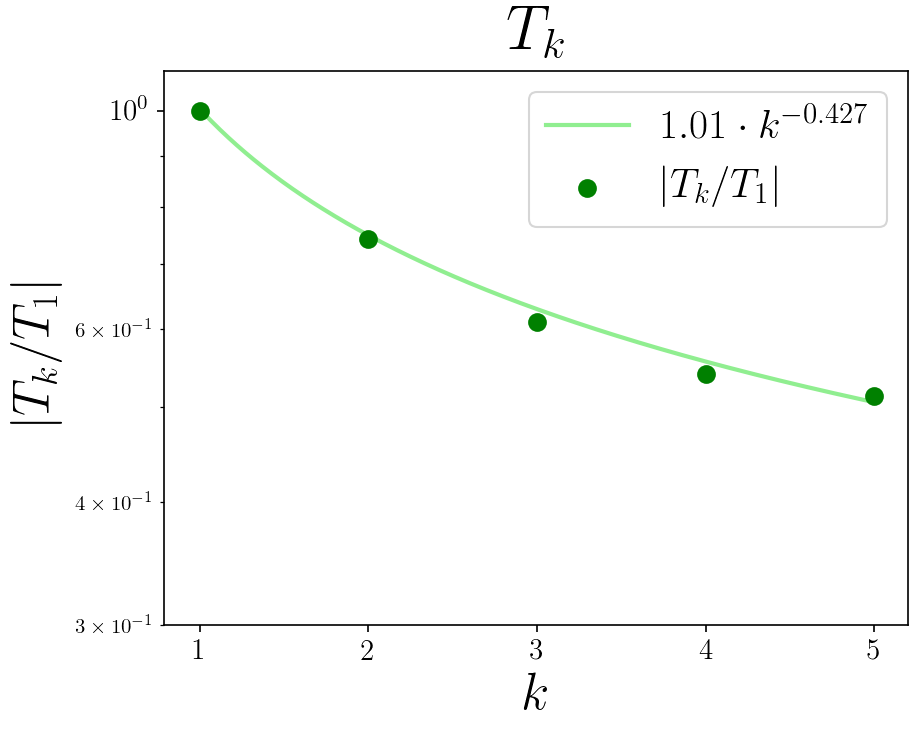}}
\caption{Spin-spin correlation function, $T_k$, in the power law decay phase on a lattice of size $N_x=10=N_y$.  This data was collected at $(c_\ell,\ c_2,\ c_4) = (3,\ 0.55,\ 0.2)$.}
\label{fig:T_k_Pow}
\end{figure}

\subsubsection{Non-integer $c_\ell$}\label{para:growing}
We also study the loop model at non-integer $c_\ell$. Although a non-integer loop weight does not correspond to a physical PEPS projector, it does correspond to a valid parameterization of the loop model. As we show below, we discover an additional phase in this regime.

In Fig. \ref{fig:eta_plot}, we display $\log(\frac{\eta}{1/4})$ for various $c_2$ and $c_\ell$. In this plot, we have fixed $c_4 = 1/15$. In addition to the phases of exponentially and power law decaying correlations, we discover a phase where the correlation functions grow with distance. As one can see in the phase plot, this phase is characterized by $\eta \approx 0$.

In Fig. \ref{fig:eta_plots}, we obtain similar phase plots for different $c_4$. We see that, at larger $c_4$, the phase of power law decay extends to smaller $c_2$ and $c_\ell$.

$ $ \\
\begin{figure}[H]
\center{\includegraphics[width=.49\textwidth]{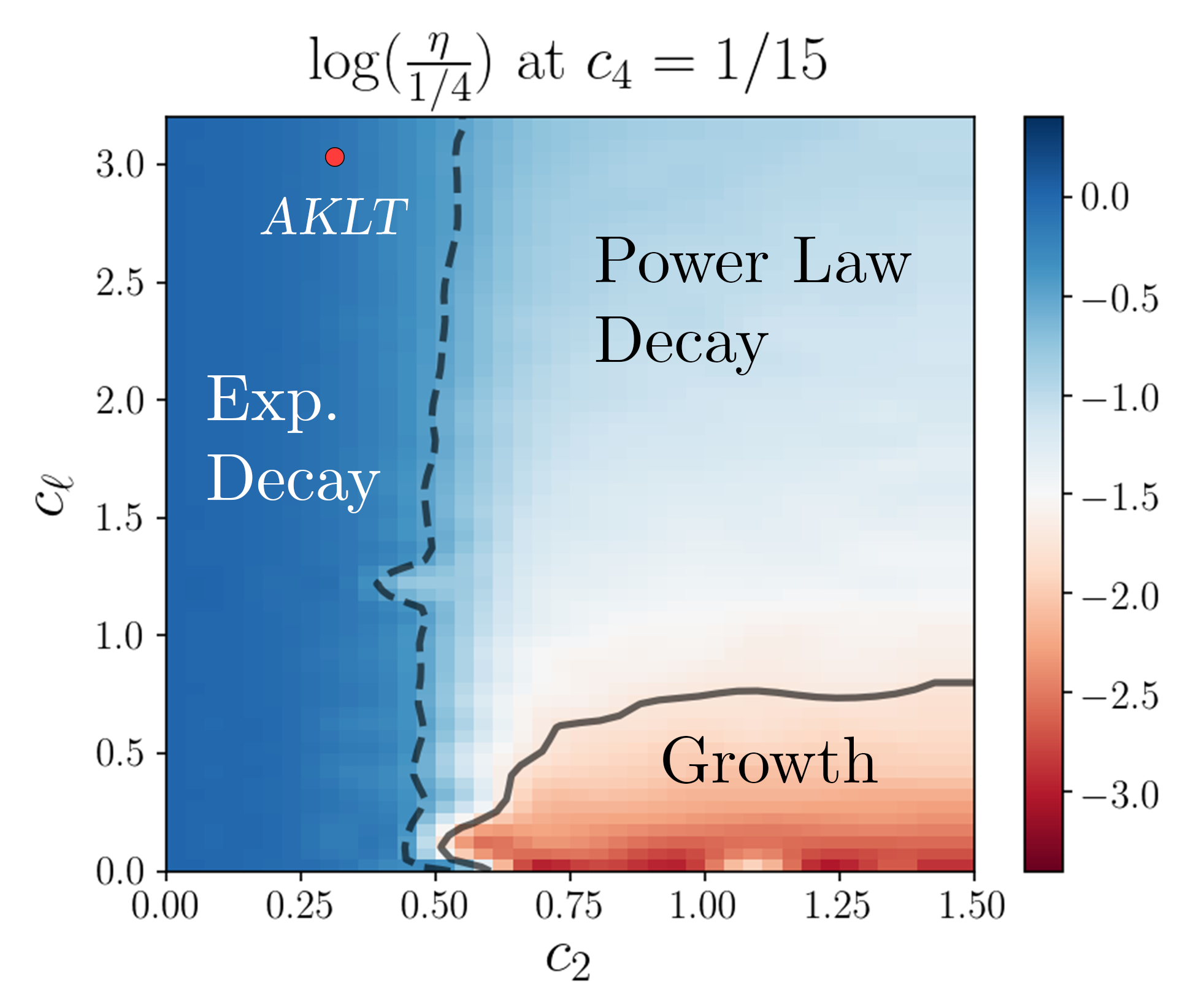}}
\caption{$\log(\frac{\eta}{1/4})$ for various $c_2$ and $c_\ell$ with fixed $c_4 = 1/15$. This plot contains the AKLT model, which we indicate with a red circle. We also label the phases of exponential decay, power law decay, and growth. To help visualize these phases, we have plotted the contour lines $\eta = 0.15$ (dashed gray line) and $\eta = 0.05$ (solid gray line).}
\label{fig:eta_plot}
\end{figure}
$ $ \\
$ $ \\
\begin{figure}[H]
\center{\includegraphics[width=.49\textwidth]{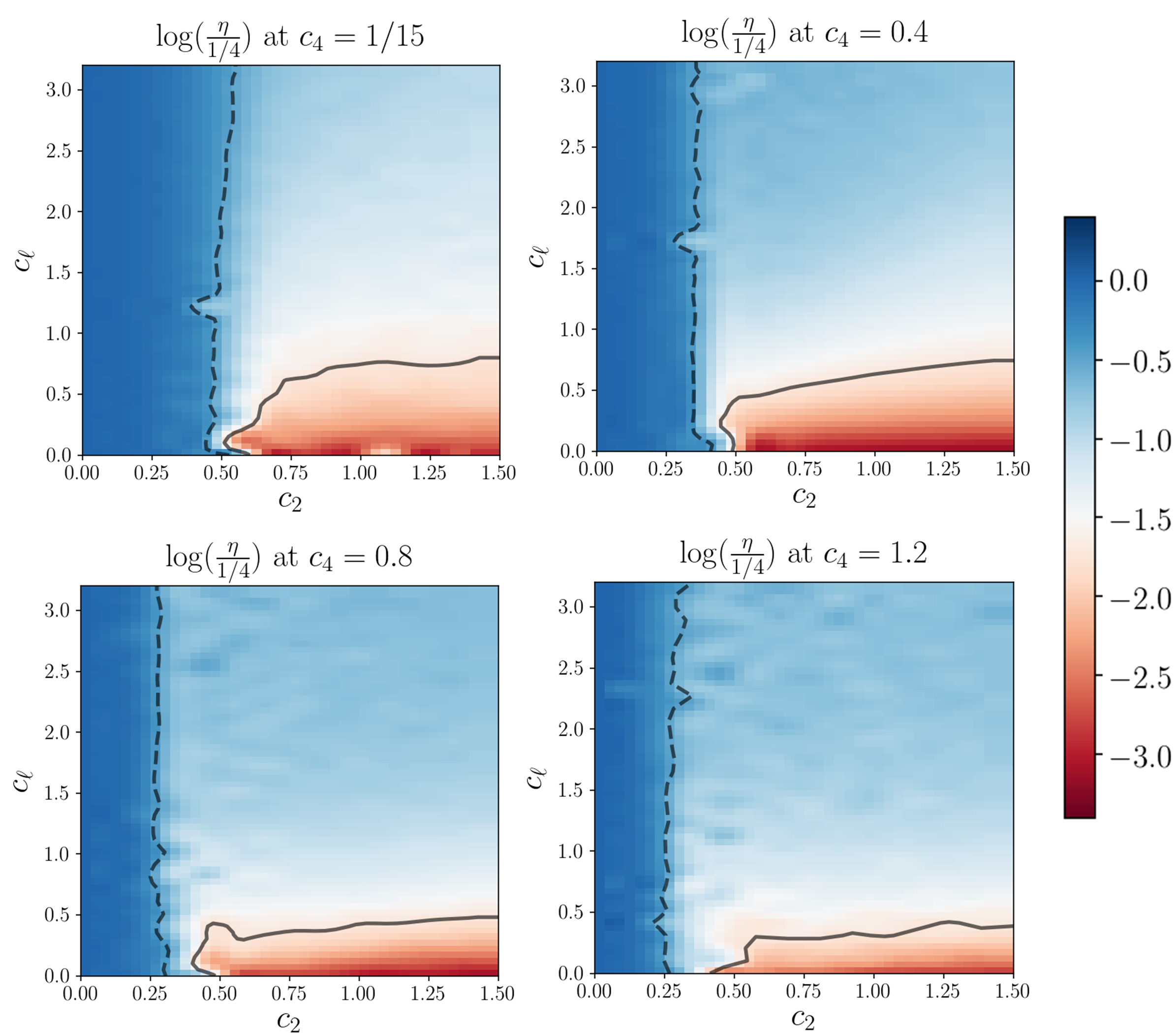}}
\caption{$\log(\frac{\eta}{1/4})$ for various $c_2$ and $c_\ell$ with fixed $c_4 $. Data is shown for $c_4 = 1/15,\ 0.4,\ 0.8,\ 1.2$. We again display the contours $\eta = 0.15$ and $\eta = 0.05$ to help visualize the phases.}
\label{fig:eta_plots}
\end{figure}
$ $ \\

\paragraph{Phase of Growing Correlations}
The phase of growing correlations occurs when the $c_2$ is sufficiently large, and the $c_\ell < 1$. This phase is characterized by a spin-spin correlation function that grows as a function of the distance between qubits, as shown in Fig. \ref{fig:T_k_Growth}. We also observe that $\eta \approx 0$ in this phase, which means that the dominant configurations contain large loops that cover many plaquettes. In addition, no physical states are contained within this phase.

\begin{figure}[!h]
\center{\includegraphics[width=.49\textwidth]{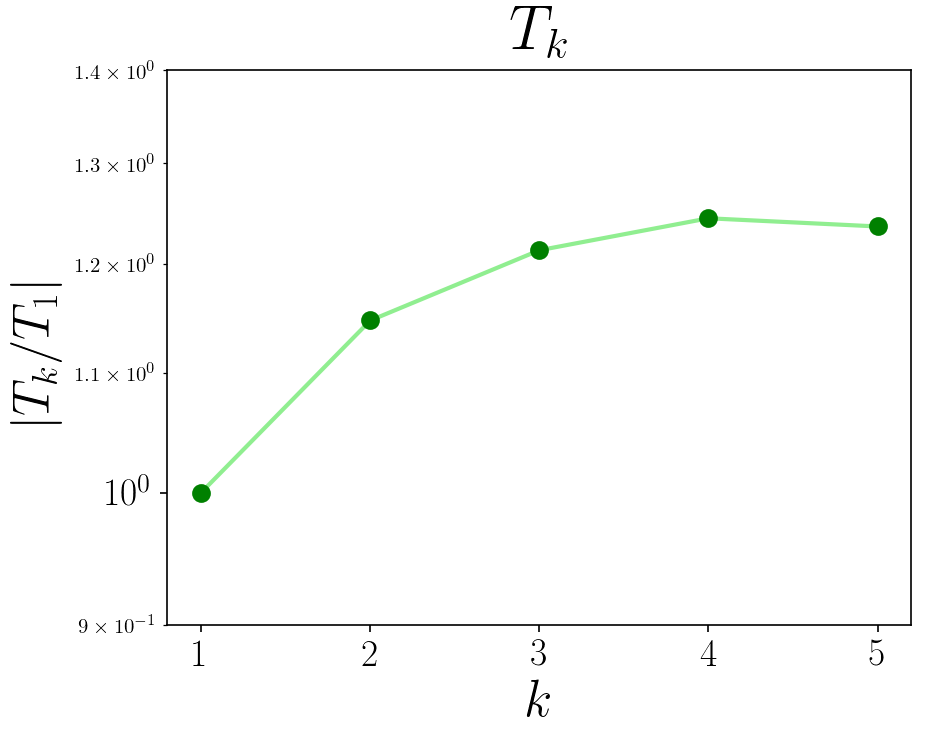}}
\caption{Spin-spin correlation function, $T_k$, in the growth phase on a lattice with $N_x=10=N_y$. This data was collected at $(c_\ell,\ c_2,\ c_4) = (0.2,\ 0.8,\ 0.2)$.}
\label{fig:T_k_Growth}
\end{figure}

\section{Discussion}
In this work, we have used the PEPS representation of the VBS to study its boundary state on a 2D square lattice. By expanding the boundary state in terms of partition functions of a classical loop model, in which loops vertices, and crossings are associated with numerical factors, we were able to use a random sampling procedure to estimate the boundary state. 
Our results indicate that the boundary Hamiltonian of the AKLT model is quasi-local, and therefore the model is gapped in the bulk.

This is moreover confirmed by perturbing the parameters of the classical loop model. We observed a phase diagram composed of a phase with exponentially decaying spin-spin correlations, a phase of power law decaying correlations, and a phase of growing correlations. Each of these phases can be diagnosed by the ratio of the average loop perimeter to the average number of loops. The point corresponding to the 2D square lattice VBS model sits inside the phase of exponential decay, confirming that is expected to be gapped in the bulk. Other PEPS can be obtained at different points of this phase diagram: the ones sitting in the same phase of the AKLT model are expected to be equally gapped, while it is unclear whether or not the phase of power law decay corresponds to gapped or gapless models.

As detailed in Appendix \ref{sec:num}, our current numerical procedure randomly samples lattice configurations to estimate partition functions. However, some refinements to this procedure are possible. It is conceivable that we may obtain better partition function estimates by using a sampling procedure that is not uniformly random, but selectively samples the dominant lattice configurations. For instance, one could use a Monte Carlo sampling procedure to obtain better estimates of the partition functions.

The approach we have taken to study the boundary states should be applicable to other PEPS models as well. If one can resolve the PEPS projection operators into Pauli operators (or some other operators with similar properties), the boundary theory can be mapped to a classical model with the help of tensor network techniques. It is hoped that this methodology can be applied to evidence bulk properties of other PEPS models by investigating their boundaries.

\textit{Acknowledgements}: J. M. thanks John Preskill and Christopher White for thoughtful discussion.
The authors acknowledge support from the Institute for Quantum Information and Matter (IQIM), an NSF Physics Frontiers Center (NSF Grant PHY-1733907). We also thank the Caltech SURF program, whose support made this work possible. 
J.~M. is supported by a Southern California Edison WAVE Fellowship. 
% K.~K. is supported by the Institute for Quantum Information and Matter (IQIM) at Caltech, which is a National Science Foundation (NSF) Physics Frontiers Center (NSF Grant PHY-1733907). 
A.~L. is supported from the Walter Burke Institute for Theoretical Physics in the form of the Sherman Fairchild Fellowship.

\bibliography{ProjectPlanRefs}
\bibliographystyle{apsrev4-1}

\appendix
\section{Construction of Projection Operators onto the Symmetric Subspace} \label{sec:projops}
In this appendix, we construct the projection operators $\symproj[d]{n}$ when $d=2$, for the cases $n=2$,  $n=3$, and  $n=4$. These are the relevant cases for the construction of the VBS on a 1D chain, on a 2D hexagonal lattice, and on a 2D square lattice, respectively.

Following Ref. \cite{2013arXiv1308.6595H}, we will denote the symmetric group on $n$ elements by $\mathcal{S}_n$. A representation of $\mathcal{S}_n$ on the vector space $(\mathbb{C}^d)^{\otimes n}$ is defined as follows. For a permutation $\pi \in \mathcal{S}_n$, we define an operator $\pi_d$ which acts on $(\mathbb{C}^d)^{\otimes n}$ and permutes inputs according to 
\begin{equation} \label{eq:pi}
\pi_d = \sum_{i_1, ..., i_n \in [1, d]} | i_{\pi^{-1}(1)}, ..., i_{\pi^{-1}(n)} \rangle \langle i_1, ..., i_n|.
\end{equation}
It follows from this definition that $\pi_d$ is a representation of $\mathcal{S}_n$ on $(\mathbb{C}^d)^{\otimes n}$. 

The symmetric subspace on $(\mathbb{C}^d)^{\otimes n}$, denoted by $\vee^n \mathbb{C}^d$, is defined as the states that are invariant under the action of $\pi_d$:
\begin{equation}
\vee^n \mathbb{C}^d = \Big\{ |\psi\rangle \in \mathbb{C}^d \ \big| \ \pi_d |\psi \rangle = |\psi\rangle \  \forall \pi_d \in \mathcal{S}_n \Big\}
\end{equation}

The projector onto $\vee^n \mathbb{C}^d$, denoted by $\symproj[d]{n}$, is given by
\begin{equation} \label{eq:proj}
\symproj[d]{n} = \frac{1}{n!} \sum_{\pi \in \mathcal{S}_n} \pi_d.
\end{equation}
This quantity is the projection operator that appears in the construction of the valence bond state.

If $d$ is small compared to $n$, the permutations operators $\pi_d$ are not linearly independent. For example, since the anti-symmetric subspace is trivial whenever $n > d$, we have that
\begin{equation} \label{eq:antiproj}
\frac{1}{n!} \sum_{\pi \in \mathcal{S}_n} \text{sgn}(\pi) \pi_d = 0.
\end{equation}
The dimension of the antisymmetric subspace is ${d \choose n}$ if $d \geq n$, and is $0$ otherwise. 

We want to show how to expand $\symproj{n}$ in terms of Pauli matrices. In order to do so, we will use the following elementary lemma.

\newtheorem{lemma}{Lemma}
\begin{lemma}
\begin{equation} \label{eq:12}
(12)_2 = \frac{1}{2}(I+{\sigma}_1 \cdot {\sigma}_2),
\end{equation} where ${\sigma} = (\sigma^x, \sigma^y, \sigma^z)$ is the vector of Pauli matrices. More generally, we have
\begin{equation} \label{eq:ij}
(ij)_2 = \frac{1}{2}(I + {\sigma}_i \cdot {\sigma}_j).
\end{equation} 
\end{lemma}

\begin{proof} Eq. (\ref{eq:pi}) dictates that 
\begin{equation}
(12)_2 = |11\rangle \langle 11| + |12\rangle \langle 21 | + |21\rangle \langle 12| + |22\rangle \langle 22|.
\end{equation}
Next, in the basis $\{|1\rangle, |2 \rangle \}$, we have
\begin{equation}
\begin{gathered}
\sigma^x_1 \otimes \sigma^x_2 = \\
\Big(|1\rangle \langle 2| + |2\rangle \langle 1| \Big) \otimes \Big(|1\rangle \langle 2| + |2\rangle \langle 1| \Big) = \\
|11\rangle \langle 22| + |21\rangle \langle 12| + |12\rangle \langle 21| + |22\rangle \langle 11|,
\end{gathered}
\end{equation}
and similarly, 
\begin{equation}
\begin{gathered}
\sigma^y_1 \otimes \sigma^y_2 = \Big(i|2\rangle \langle 1| -i |1\rangle \langle 2| \Big) \otimes \Big(i|2\rangle \langle 1| -i |1\rangle \langle 2| \Big) = \\
-|221\rangle \langle 11| + |21\rangle \langle 12| + |12\rangle \langle 21| - |11\rangle \langle 22|,
\end{gathered}
\end{equation}
and lastly, 
\begin{equation}
\begin{gathered}
\sigma^z_1 \otimes \sigma^z_2 = \Big(|1\rangle \langle 1| - |2\rangle \langle 2| \Big) \otimes \Big(|1\rangle \langle 1| - |2\rangle \langle 2| \Big) = \\
|11\rangle \langle 11| - |21\rangle \langle 21| - |12\rangle \langle 12| + |22\rangle \langle 22|.
\end{gathered}
\end{equation}
Combining these results with $I = |11\rangle\langle 11| + |21\rangle \langle 12| + |12 \rangle \langle 12| + |22\rangle \langle 22|$, we have
\begin{equation}
\begin{gathered}
I + {\sigma}_1 \cdot {\sigma}_2 = \\ 2\Big(|11\rangle \langle 11| + |21\rangle \langle 12| + |12 \rangle \langle 21| + |22 \rangle \langle 22| \Big) = \\ 2 (12)_2.
\end{gathered}
\end{equation}
It follows from this relation that we can write arbitrary two cycles in a $d=2$ representation as 
\begin{equation}
(ij)_2 = \frac{1}{2} (I + {\sigma}_i \cdot {\sigma}_j ).
\end{equation}
\end{proof}

Since every permutation can be written as a product of two cycles, this lemma alone implies that we can expand $\symproj{n}$ as a polynomial in the terms $\{ \sigma_i \cdot \sigma_j\}$. In what follows, we will use the linear dependence of the permutations to further simplify the expression.

\subsection{\texorpdfstring{$n=2$}{n=2}} \label{d2n2}
In the case $n=2$, the symmetric group consists of $\mathcal{S}_2 = \{ \text{id}, (12) \}$.
We can expand $\symproj{2}$ via equation \ref{eq:proj} as:
\begin{equation}
\begin{gathered}
\symproj{2} = \frac{1}{2}\Big( \text{id}_2 + (12)_2 \Big) = \\ \frac{1}{2}\Big( I + \frac{1}{2} (I+{\sigma}_1\cdot {\sigma}_2) \Big) = \frac{3}{4} \Big( I + \frac{1}{3}{\sigma}_1\cdot {\sigma}_2 \Big). \qedhere
\end{gathered}
\end{equation}
This is the projection operator used in the construction of the 1D VBS. 

\subsection{\texorpdfstring{$n=3$}{n=3}} \label{d2n3}
In the case $n=3$, we note that the dimension of the antisymmetric space is $0$, and hence Eq. (\ref{eq:antiproj}) implies that
\begin{equation}
(12)_2 + (23)_2 + (13)_2 = \text{id}_2 + (123)_2 + (132)_2.
\end{equation} 
Using this relationship along with Eq. (\ref{eq:ij}), we obtain 
\begin{equation}
\begin{gathered}
\symproj{3} = \frac{1}{3!} \sum_{\pi \in \mathcal{S}_3} \pi_2 = \\
\frac{1}{6}\Big( \text{id}_2 + (12)_2 + (23)_2 + (13)_2 + (123)_2 + (132)_2 \Big) = \\
\frac{1}{3}\Big((12)_2 + (23)_2 + (13)_2\Big) = \\
\frac{1}{2} \Big(I+\frac{1}{3} ({\sigma}_1\cdot {\sigma}_2 + {\sigma}_2\cdot {\sigma}_3  + {\sigma}_1\cdot {\sigma}_3 ) \Big).
\end{gathered}
\end{equation}
This is the projection operator used in the construction of the VBS on a 2D hexagonal lattice. 

\subsection{\texorpdfstring{$n=4$}{n=4}} \label{d2n4}
In the case $n=4$, we again have that the antisymmetric space is empty, which implies that 
\begin{equation}
\sum_{\pi \in \mathcal{S}_4,\ \text{sgn}(\pi) = +1} \pi_2 = \sum_{\pi \in \mathcal{S}_4,\ \text{sgn}(\pi) = -1} \pi_2.
\end{equation} 
Thus, we can construct $\symproj{4}$ as a sum of only even or odd permutations in $\mathcal{S}_4$. We choose the even permutations, which consist of the identity, the $3$-cycles, and the $(2,2)$-cycles:
\begin{equation}
\begin{gathered}
\symproj{4} = \frac{2}{4!} \sum_{\pi \in \mathcal{S}_4,\ \text{sgn} = +1} \pi_2 = \\
\frac{1}{12} \Bigg( \text{id}_2 + \sum_{(ijk) \in \mathcal{S}_4} (ijk)_2 + \sum_{(ij)(kl) \in \mathcal{S}_4} (ij)_2(kl)_2 \Bigg),
\end{gathered}
\end{equation}
where $\sum_{(ijk) \in \mathcal{S}_4}$ is a sum over all $3$-cycles (there are 8 of them), and $\sum_{(ij)(kl) \in \mathcal{S}_4}$ is a sum over all disjoint $(2,2)$-cycles (there are 3 of them). 

To evaluate the $3$-cycles, note that we can write a $3$-cycle as a product of $2$-cycles: $(ijk) = (jk)(ij)$. Similarly, this cycle's inverse can be written as $(kji) = (ij)(jk)$. Therefore, using Eq. (\ref{eq:ij}), we can write the sum of this cycle and its inverse in the $d=2$ representation as 
\begin{equation}
\begin{gathered}
(ijk)_2 + (kji)_2 = (ij)_2(jk)_2 + (jk)_2(ij)_2 = \\
\frac{1}{2}(I + {\sigma}_i \cdot {\sigma}_j + {\sigma}_j \cdot {\sigma}_k + {\sigma}_i \cdot {\sigma}_k). 
\end{gathered}
\end{equation}
With this relation, we can sum over all $3$-cycles by summing over the 4 distinct cycle-and-inverse combinations, revealing that
\begin{equation}
\sum_{(ijk) \in \mathcal{S}_4} (ijk)_2 = 2I + \sum_{(ij)\in \mathcal{S}_4} {\sigma}_i \cdot {\sigma}_j.
\end{equation} 

Next, the sum over $(2,2)$-cycles are simple because $(ij)_2$ commutes with $(kl)_2$:
\begin{equation}
\begin{gathered}
\sum_{(ij)(kl) \in \mathcal{S}_4} (ij)_2(kl)_2 = \\ \frac{3}{4} I + \frac{1}{4}\sum_{(ij) \in \mathcal{S}_4} ({\sigma}_i \cdot {\sigma}_j) + \frac{1}{4} \sum_{(ij)(kl) \in \mathcal{S}_4} ({\sigma}_i \cdot {\sigma}_j) ({\sigma}_k \cdot {\sigma}_l).
\end{gathered}
\end{equation}

Combining all of these relations, we have our final result:
\begin{align}
\symproj{4} = \frac{15}{48} \Bigg(I &+ \frac{1}{3}\sum_{(ij) \in \mathcal{S}_4} ({\sigma}_i \cdot {\sigma}_j) \notag \\ &+ \frac{1}{15}\sum_{(ij)(kl) \in \mathcal{S}_4} ({\sigma}_i \cdot {\sigma}_j)({\sigma}_k \cdot {\sigma}_l) \Bigg).
\end{align}
This projection operator is used in the construction of the VBS on a 2D square lattice.

\section{Numerical Techniques} \label{sec:num}
\subsection{Setup}
In Sec. \ref{sec:map}, we showed how the boundary state of the VBS can be computed from partition functions of a classical loop model. In this appendix, we describe our procedure for numerically estimating these partition functions and constructing the boundary state.

We begin by outlining our setup. Our goal is to study the boundary state of the VBS on a 2D square lattice. We will study square lattices that are periodic in one direction - i.e. cylinders. In particular, we will consider a half-open cylinder, where strings start and end on only one side of the cylinder, effectively reducing the boundary sites to only one side. Let there by $N_x$ boundary spins in the periodic direction (or equivalently, $N_x$ plaquettes in this direction), and $N_y-1$ spins in the longitudinal direction (or equivalently, $N_y$ plaquettes). This lattice is pictured in Fig. \ref{fig:square_lattice}.

\begin{figure}[htb]
\center{\includegraphics[width=.3\textwidth]{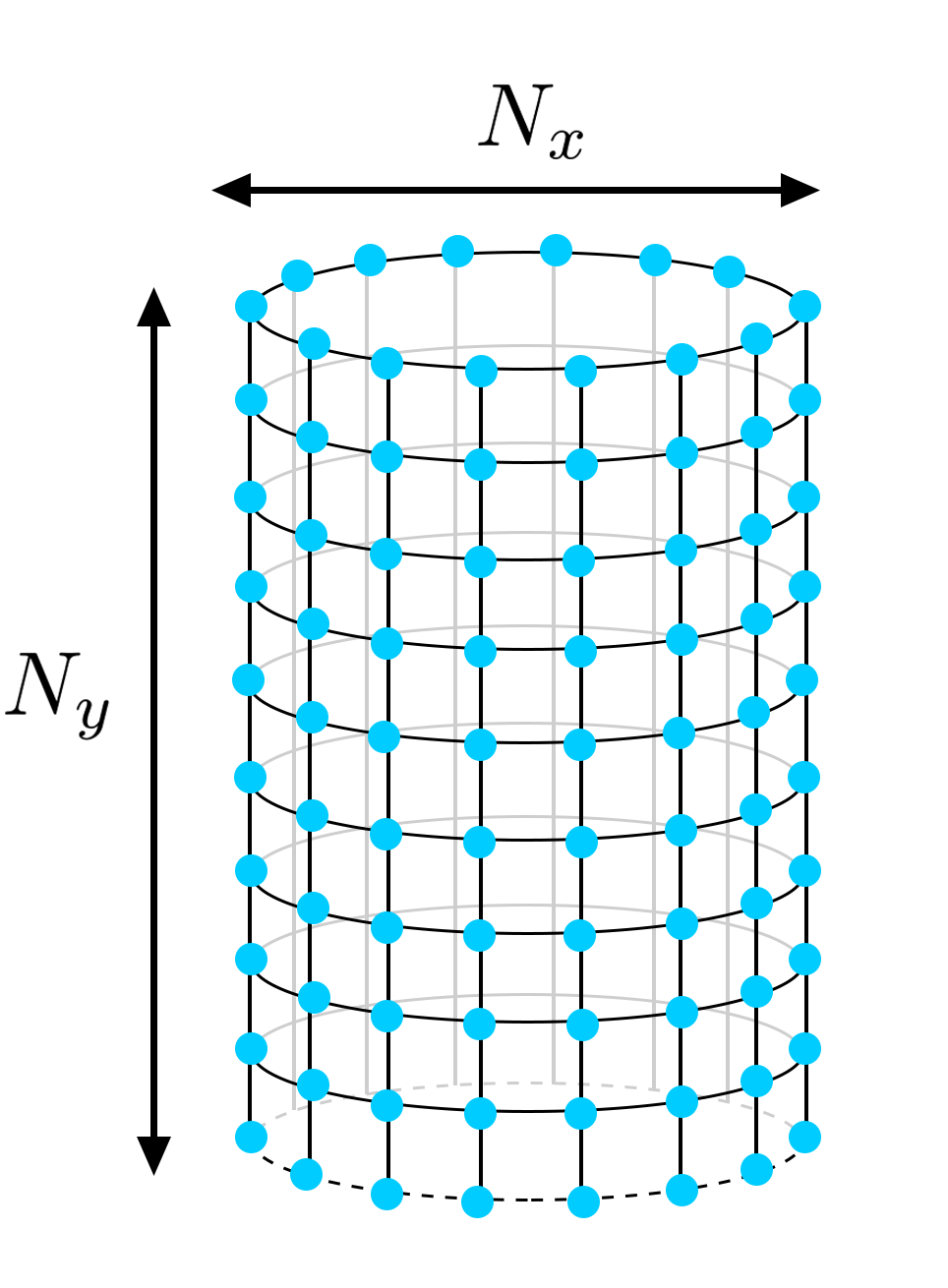}}
\caption{The cylindrical square lattice on which we analyze the VBS and its boundary state. There are $N_x$ plaquettes in the periodic direction, and $N_y$ plaquettes in the longitudinal direction. The dashed lines on the bottom spins signify that this is a half-open cylinder.}
\label{fig:square_lattice}
\end{figure}

It is justified to study the AKLT model on the half-open cylinder rather than a full cylinder because correlations between the two sides of a full cylinder decay exponentially with cylinder length, and hence vanish in the thermodynamic limit. We show this graphically in Figure \ref{fig:T_N_y}. In this figure, we plot as a function of $N_y$ the spin-spin correlation functions of qubits on opposite sides of a full cylinder (one qubit at the top and one at the bottom). We clearly see that the correlation functions decay exponentially with the cylinder length, and therefore the two sides of the full cylinder are uncorrelated in the thermodynamic limit. Thus, the results we obtain with the half-open cylinder in the thermodynamic limit should be identical to those of a full cylinder.

\begin{figure}[htb]
\center{\includegraphics[width=.48\textwidth]{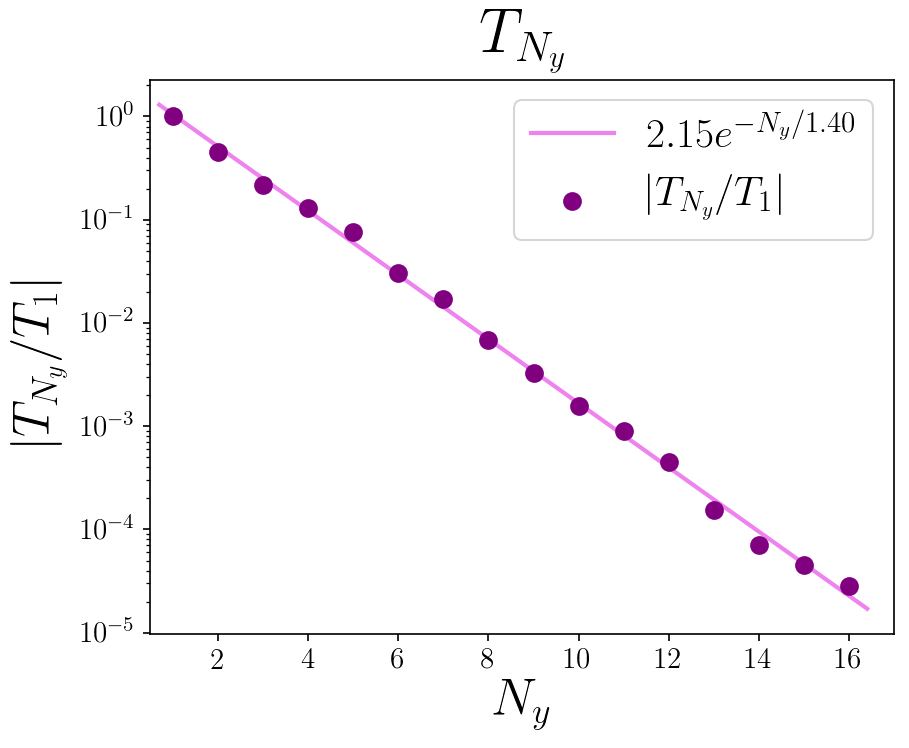}}
\caption{Correlation function of spins on opposite sides of a cylinder with $N_x = 10$. The correlation function decays exponentially with $N_y$, indicating that the two sides of the cylinder are uncorrelated with each other in the thermodynamic limit.}
\label{fig:T_N_y}
\end{figure}

\subsection{Random Sampling}
To compute the boundary state of our lattice, we must calculate the coefficients $\{ Z^{\textbf{i}}\}_\textbf{i}$. As each one of these coefficients is a partition function, we can compute it by summing over the weights of all allowed configurations with boundary conditions specified by $\textbf{i}$. However, a simple counting argument indicates that the number of such configurations grows exponentially in the total number of plaquettes on the lattice: $N = N_x N_y$. Summing over all of these configurations is clearly intractable for large lattices, so we turn to sampling techniques to estimate the coefficients $\{ Z^{\textbf{i}}\}_\textbf{i}$. 

We have written up codes in Python to estimate these coefficients and construct an approximate boundary state. These codes are available on \href{https://github.com/jmmartyn/AKLT_Boundary_Theory}{Github}.

We now outline the procedure we used. We sample different configurations by flipping plaquettes. To flip a plaquette, invert the edges of the spins that border the plaquette: if there is no string, add one; if there is a string, remove it. We depict an elementary plaquette flip in Fig. \ref{fig:plaquette_flip}. As we describe below, we use this plaquette flipping to sample all possible configurations $\Delta \in \mathcal{C}_{\textbf{i}}$.

\begin{figure}[htb]
\center{\includegraphics[width=.48\textwidth]{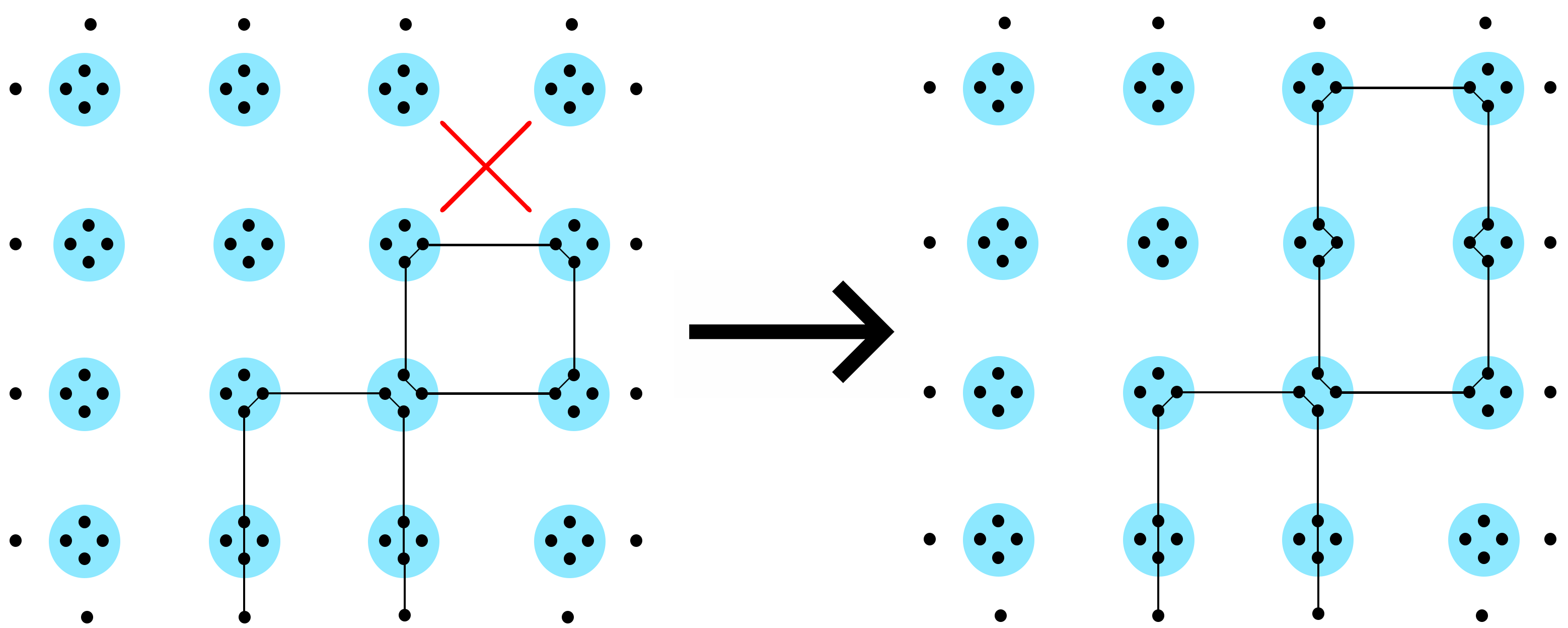}}
\caption{A plaquette flip performed on the plaquette in the upper right hand corner.}
\label{fig:plaquette_flip}
\end{figure}

Next, note that the partition function can be split into a sum of contributions with $n$ plaquettes flipped. Mathematically,
\begin{equation}
Z^{\textbf{i}} = \sum_{n=0}^{N} Z_n^{\textbf{i}},
\end{equation}
where $Z_n^{\textbf{i}}$ is the contribution to the partition function with $n$ plaquettes flipped. 

We can estimate $Z^{\textbf{i}}$ to within an error $\epsilon$ by truncating the above summation below at some $n_{\text{low}}$ and above at some $n_{\text{high}}$:
\begin{equation}
Z^{\textbf{i}} = \sum_{n=0}^{N} Z_n^{\textbf{i}} \approx \sum_{n=n_{\text{low}}}^{n_{\text{high}}} Z_n^{\textbf{i}}
\end{equation}
Essentially, the majority of $Z^{\textbf{i}}$ is contained in $\{ Z_n^{\textbf{i}} \ |\  n \in [n_{\text{low}}, \ n_{\text{high}}] \ \}$. For instance, the greatest contributions to $Z^{\emptyset}$ comes from the configurations with few vertices of degree 2 or 4 because the factors associated with these vertices, $\frac{1}{3}$ and $\frac{1}{15}$, respectively, decrease the magnitude of the weight. Because flipping plaquettes increases the number of vertices with degree 2 or 4, we can estimate $Z^{\emptyset}$ by retaining only the first few $Z_n^{\emptyset}$. In our method, we determine $n_{\text{low}}$ and $n_{\text{high}}$ by (approximately) calculating the values of $Z^{\emptyset}$ until the maximum contribution from terms outside of $n\in [n_{\text{low}}, \ n_{\text{high}}]$ are less than $\epsilon$. In practice, when computing the coefficients for the AKLT model, we found that typically $n_{\text{low}} = 0$ and $n_{\text{high}} \approx N_x N_y /4$ was sufficient for $\epsilon \lesssim 0.01$. When $c_2$ and $c_4$ slightly increase beyond their values at the AKLT model (for instance, to $c_2 = 0.55$, $c_4 = 0.2$), we find that $n_{\text{low}}$ and $n_{\text{high}}$ increase to some $n_{\text{low}} > 0$ and $n_{\text{high}} \sim N_xN_y/2$. When $c_2$ and $c_4$ are increased further to values of order 1, using $n_{\text{low}}$ and $n_{\text{high}}$ is not very effective, as typically all $Z_n^{\textbf{i}}$ are required for an accurate estimation of $Z^\textbf{i}$.

We then approximate each $Z_n^{\textbf{i}}$ with random sampling. First, we begin with the simplest possible configuration that has boundary conditions specified by $\textbf{i}$. For instance, the simplest possible configuration contributing to $Z^{\emptyset}$ is the lattice with no strings and no loops. Likewise, the simplest configuration corresponding to $Z^{\{(i_1, j_1)\}}$ is the configuration with the shortest possible string connecting vertices $i_1$ and $j_1$. Then, for $N_{\text{samples}}$ times at a fixed $n$, we randomly flip $n$ plaquettes and each time arrive at a new configuration. 

However, note that flipping plaquettes cannot generate all possible configurations of vertices with two strings passing through (vertices with degree 4); for instance, flipping plaquettes cannot generate a crossing of two strings. To account for all possible configurations, we begin with a configuration with $n$ plaquettes flipped, and randomly sample over $\leq N_{v_4 \text{ samples}}$ possible configurations of vertices with degree 4 that arise in our lattice. For each configuration, we check that it has the correct boundary conditions and then compute its weight, $w(\Delta)$. We average over the weights of the $\sim N_{v_4 \text{ samples}}$ configurations and multiply by the total number of possible configurations of degree 4 vertices. This produces the expected value of the sum of weights corresponding to our particular choice of plaquette flips. 

We then average each of these (sums of) weights over the $N_{\text{samples}}$ choices of plaquette flips to determine the average weight of all configurations with $n$ plaquettes flipped, which we denote by $\bar{w}_{n}^{\textbf{i}}$. Since there are ${N}\choose{n}$ possible configurations with $n$ plaquettes flipped, the contribution to the partition function with $n$ plaquettes flipped is approximately
\begin{equation}
Z_n^{\textbf{i}} \approx \bar{w}_{n}^{\textbf{i}} \cdot {{N}\choose{n}}.
\end{equation}

There is one more caveat before using these $Z_n^{\textbf{i}}$'s to estimate $Z^{\textbf{i}}$. On a cylindrical lattice, configurations of strings can contain noncontractible loops, which cannot be generated in odd numbers by flipping plaquettes. To include the contributions from these noncontractible loops, we repeat the above procedure, but now, when flipping plaquettes, we also flip all the edges at some rung $r \in [1, N_y+1]$, effectively adding a noncontractible loop to the configuration. On average, combining this with our previous procedure samples over all configurations of noncontractible loops. We then add to $Z_n^{\textbf{i}}$ the corresponding contribution from configurations with an additional noncontractible loop (i.e. $Z_n^{\textbf{i}} \mapsto Z_{n}^{\textbf{i}} + Z_{n,\ \text{noncontractible}}^{\textbf{i}}$).

Finally, we compute $Z^{\textbf{i}} \approx \sum_{n=n_{\text{low}}}^{n_{\text{high}}} Z_n^{\textbf{i}}$ to obtain an approximation to $Z^{\textbf{i}}$. To reduce statistical fluctuations, we iterate this entire procedure $N_{\text{iterations}}$ times and take the average value of $Z^{\textbf{i}}$ as our best estimate of the partition function. From these values, the boundary state was constructed. For small $c_2$, it was observed that partition functions corresponding to many strings were exponentially small and only weakly perturbed the properties of the boundary state; some partition functions could be neglected while maintaining the properties boundary state to within a small error, say $< $ 1\%. Evidently then, we can compute the coefficients $\{Z^{\textbf{i}}\}_\textbf{i}$ and the boundary state to arbitrary accuracies by decreasing $\epsilon$ and increasing $N_{\text{samples}}$, $N_{v_4 \text{ samples}}$, and $N_{\text{iterations}}$. 
$ $ \\

\subsection{Symmetry Considerations}
We can simplify the computation of $\rho_\partial$ by evoking the rotational symmetry of a cylindrical lattice in its periodic direction. For instance, due to this symmetry, $Z^{\{(i, j)\}}$ only depends on the distance between vertices $i$ and $j$ on a a circle of length $N_x$. For instance, $Z^{\{(1, 2)\}} = Z^{\{(2, 3)\}} = Z^{\{(1, N_x)\}}$, and so on. This symmetry persists to higher order partition functions as well.

In our computation of the boundary state, we only estimate one partition function from each conjugacy class of coefficients related by this symmetry. We then set all coefficients in this class equal to this estimated partition function. For example, we would estimate $Z^{\{(1, 2)\}}$, and then set $Z^{\{(i, i+1 (\text{mod} N_x))\}} = Z^{\{(1, 2)\}}$ for all $i \in [1, N_x]$. This significantly reduces the number of necessary computations and ensures that the approximate boundary state exhibits translational symmetry.

\end{document}